\documentclass[pdflatex,sn-mathphys-num]{sn-jnl}
\usepackage{graphicx}%
\usepackage{multirow}%
\usepackage{amsmath,amssymb,amsfonts}%
\usepackage{amsthm}%
\usepackage{mathrsfs}%
\usepackage[title]{appendix}%
\usepackage{xcolor}%
\usepackage{textcomp}%
\usepackage{manyfoot}%
\usepackage{booktabs}%
\usepackage{algorithm}%
\usepackage{algorithmic}
\usepackage{listings}%
\usepackage{bm}

 \usepackage{subcaption}  
 \usepackage{float}  
 \usepackage{placeins}

\usepackage{tikz}
\usetikzlibrary{arrows.meta} 

\usepackage{afterpage}

\newtheorem{Lem}{\textbf{Lemma}}
\newtheorem{Def}{\textbf{Definition}}
\newtheorem{Rmk}{\textbf{Remark}}

\newtheorem{Cor}{\textbf{Collary}} 

\newtheorem{Prop}{\textbf{Proposition}}

\renewcommand{\bf}{\mathbf} 
\renewcommand{\rm}{\mathrm} 

\theoremstyle{thmstyleone}
\theoremstyle{thmstyletwo}%
\theoremstyle{thmstylethree}%


\begin{document}

\title[Article Title]{Moving Obstacle Collision Avoidance via Chance-Constrained Model Predictive Control with Control Barrier Function}
\author[1]{Ming Li}\email{ming3@kth.se} 
\author*[2]{Zhiyong Sun}\email{zhiyong.sun@pku.edu.cn} %
\author[3]{Zirui Liao}\email{by2003110@buaa.edu.cn}
\author[4]{Siep Weiland}\email{s.weiland@tue.nl}

\affil[1]{
  \orgdiv{Division of Decision and Control Systems}, 
  \orgname{KTH Royal Institute of Technology}, 
  \orgaddress{
    \city{Stockholm}, 
    \postcode{100 44}, 
    \country{Sweden}}}
\affil[2]{
  \orgdiv{Department of Mechanics and Engineering Science}, 
  \orgname{College of Engineering}, 
  \orgname{Peking University}, 
  \orgaddress{
    \city{Beijing}, 
    \postcode{100871}, 
    \state{Beijing}, 
    \country{China}}}
\affil[3]{
  \orgdiv{School of Automation Science and Electrical Engineering}, 
  \orgname{Beihang University}, 
  \orgaddress{
    \city{Beijing}, 
    \postcode{100191}, 
    \state{Beijing},                
    \country{China}}}
\affil[4]{
  \orgdiv{Department of Electrical Engineering}, 
  \orgname{Eindhoven University of Technology (TU/e)}, 
  \orgaddress{
    \city{Eindhoven}, 
    \postcode{310024}, 
    \state{5600 MB},             
    \country{The Netherlands}}}

\abstract{Model predictive control (MPC) with control barrier function (CBF) is a promising solution to address the moving obstacle collision avoidance (MOCA) problem. Unlike MPC with distance constraints (MPC-DC), this approach facilitates early obstacle avoidance without the need to increase prediction horizons. However, the existing MPC-CBF method is deterministic and fails to account for perception uncertainties. This paper proposes a generalized MPC-CBF approach for stochastic scenarios, which maintains the advantages of the deterministic method for addressing the MOCA problem. Specifically, the chance-constrained MPC-CBF (CC-MPC-CBF) technique is introduced to ensure that a user-defined collision avoidance probability is met by utilizing probabilistic CBFs. However, due to the potential empty intersection between the reachable set and the safe region confined by CBF constraints, the CC-MPC-CBF problem can pose challenges in achieving feasibility. To address this issue, we propose a sequential implementation approach that first solves a standard MPC optimization problem, followed by a predictive safety filter optimization. The safety filter is handled using a novel iterative convex optimization algorithm. This sequential approach improves feasibility compared to the CC-MPC-CBF optimization although it sacrifices stability performance. We apply our proposed algorithm to a double integrator system for MOCA, and we showcase its robustness to obstacle measurement uncertainties and favorable feasibility properties.}




\maketitle

\section{Introduction}\label{sec1}

Moving Obstacle Collision Avoidance (MOCA) aims to enable objects to navigate safely in dynamic environments with moving obstacles. However, this critical issue presents a difficult challenge for ensuring safety in dynamic and unpredictable environments across numerous applications, including robotic arms~\cite{Manipulator_Control}, autonomous vehicles~\cite{Autonomous_Vehicles}, and formation control~\cite{Formation_Control}. To address this challenge, a range of methods have been developed, such as velocity obstacles~\cite{Obstacle_Velocity}, potential fields~\cite{Potential_fields}, and model predictive control (MPC)~\cite{MPC_1,MPC_2,MPC_CBF}. Among the various solutions, MPC is a favored approach for solving the MOCA problem due to its capability to predict the future behavior of the system and optimize a control action to satisfy safety constraints.

As noted in~\cite{MPC_CBF}, the existing literature on MPC for addressing the MOCA problem commonly relies on Euclidean norms to formulate distance constraints. This approach causes robots to delay obstacle avoidance until obstacles are very close, activating the control process only once the reachable set intersects with obstacles. As a result, robots often produce aggressive corrective maneuvers and oscillatory trajectories. To mitigate this reactive behavior, recent studies have proposed several alternative methods. These include: i) Geometry-aware metrics.  This approach replaces the Euclidean norm with weighted Mahalanobis or anisotropic ellipsoidal distances, which allow a more accurate representation of robot–obstacle geometry and sensing anisotropies~\cite{WeightedDist_IROS19, MahalanobisCBF2022, AnisotropicCBF_TRO22}; ii) Velocity-dependent safety margins. In this approach, it expands collision avoidance distances using speed-scaled buffers or time-to-collision constraints, enabling earlier hazard anticipation and response \cite{SpeedAwareSMC18, Faessler2018SpeedAware}; and iii) MPC with Control Barrier Function (CBF), known as MPC-CBF \cite{MPC_CBF}. By combining predictive optimization and velocity-aware safety envelopes, MPC-CBF proactively avoids collisions earlier than traditional distance-constrained MPC or greedy, instantaneous CBF-QP. In this paper, we advocate for the MPC-CBF approach because it effectively integrates the predictive capabilities of MPC with the formal safety guarantees of CBFs, which generate smooth and provably safe trajectories with reduced conservatism.

However, a notable limitation of the MPC-CBF formulation in \cite{MPC_CBF} is its assumption of perfect knowledge of system dynamics. In practical applications, uncertainties in system dynamics, environment perception, and modeling errors are inevitable and can significantly degrade controller performance or even lead to unsafe behaviors. Addressing these uncertainties within the MPC-CBF framework is therefore crucial for ensuring robust and reliable operation. For the MPC-CBF formulation proposed in~\cite{MPC_CBF} (which ignores uncertainty), it can be extended to handle both bounded and unbounded  uncertainties by leveraging recent advances in robust, adaptive, data-driven, and stochastic CBF theory. For bounded (deterministic) uncertainties, robust and ISSf-CBFs explicitly subtract an infinity norm bound from the CBF derivative condition to create an invariant tube around the nominal trajectory, enabling tube-based robust MPC-CBF schemes~\cite{Robust_CBF1,Robust_CBF2,ISSf,Tube_Robust_MPC_CBF}. Alternatively, adaptive CBFs tighten the barrier constraint online as parameter estimates converge and guarantee safety once a finite-time excitation condition is met \cite{adactive_cbf,xiao2021adaptive}. Data-driven or Bayesian variants use Gaussian-process models to wrap the CBF inequality in a high-probability confidence bound that shrinks as more data arrive, yielding statistical robustness when combined with predictive safety certification~\cite{Data_driven_CBF,GP_CBF_Bayesian}. When uncertainties are unbounded, stochastic CBF methods offer valuable insights by generalizing barrier conditions to stochastic (Itô) dynamics. This is achieved by constraining the Kolmogorov backward operator, enabling finite-horizon MPC-CBF formulations that directly control the probability of safety violations~\cite{Stochastic_CBF}. Additionally, these principles have been integrated with learning-based approaches: for example, safe reinforcement learning methods filter exploratory actions via real-time CBF~\cite{RL_CBF_SRL}, and model-predictive safety shields employ minimal-intervention MPC-CBF backups activated only when nominal policies approach unsafe boundaries \cite{MPSC_CBF}.  Alternatively, chance-constrained MPC~\cite{schwarm1999chance} can be integrated with CBF to enforce probabilistic safety guarantees. While recent studies have begun to explore this integration under unbounded uncertainties~\cite{wang2024stochastic,liu2025flexible}, several important issues remain insufficiently addressed. These include the ability to ensure early collision avoidance, as well as a thorough understanding of how key design parameters influence feasibility, computational efficiency, and safety assurances.

In this paper, we investigate the MOCA problem in stochastic scenarios where obstacle measurements are subject to unbounded uncertainties. Specifically, we extend the MPC-CBF framework by incorporating chance-constrained formulations. Our goal is to generalize the results from~\cite{MPC_CBF}, originally developed for deterministic settings, to more realistic stochastic environments and to provide a comprehensive analysis of feasibility, computational performance, and safety guarantees. The main contributions of this paper are three-fold:
\begin{enumerate}
    \item We propose the CC-MPC-CBF to address the MOCA problem in a stochastic environment, which extends the (deterministic) MPC-CBF solution presented in~\cite{MPC_CBF}. The CC-MPC-CBF approach combines MPC with chance-constrained CBFs to handle stochastic uncertainties and provide probabilistic guarantees of safety. This approach allows for some violations of the CBF constraint with a specified probability over an infinite time horizon while ensuring that the probability of collision is below a user-specified threshold.  
    \item We develop a sequential implementation approach to address the CC-MPC-CBF to improve the feasibility, which includes two sub-optimization problems, i.e., a standard MPC and a predictive safety filter. The predictive safety filter is addressed using a novel iterative convex optimization algorithm that provides solutions efficiently. By utilizing this methodology, we can obtain comparable performances as the CC-MPC-CBF with improved feasibility and fast computation speed.
    \item We demonstrate the effectiveness of our developed algorithms by applying them to a real MOCA example, and showcase their advantageous properties through numerous simulation results. Specifically, we highlight that our approach is robust to stochastic sensing uncertainties in obstacle measurements, achieves a high success rate for MOCA, and is feasible for real-world applications.
\end{enumerate}
The remainder of this paper is organized as follows. Section~\ref{Preliminaries} introduces the necessary mathematical preliminaries and presents the problem formulation considered in this paper. In Section~\ref{Main_results}, we develop the CC-MPC-CBF framework and analyze how key design parameters influence feasibility, computational efficiency, and safety guarantees. Section~\ref{Sequantial_Approach} describes a sequential implementation of the proposed method and introduces an iterative convex optimization algorithm for efficient solution of the chance-constrained problem. In Section~\ref{Simulation}, we validate the proposed framework through a simulation of a double integrator system in a MOCA setting. Finally, Section~\ref{Conclusions} concludes the paper.
\section{Preliminaries and Problem Statement}\label{Preliminaries} 
In this section, we provide an overview of the system models and definitions of discrete-time CBFs. We also examine the limitations of MPC with CBF and its inability to address system uncertainties. This motivates the chance-constrained formulation of MPC-CBF.
\subsection{System Models}
Consider the following model governing the motion of the robot.	
\begin{equation}\label{Affine_Control_System}
	    \mathbf{x}_{k+1}=\mathbf{f}(\mathbf{x}_{k})+\mathbf{g}(\mathbf{x}_{k})\mathbf{u}_{k},
\end{equation}
where $\mathbf{x}_{k}\in\mathcal{X}\subset\mathbb{R}^{n}$, $\mathbf{u}_{k}\in\mathcal{U}\subset\mathbb{R}^{m}$, and the functions $\mathbf{f}:\mathbb{R}^{n}\rightarrow\mathbb{R}^{n}$ and $\mathbf{g}:\mathbb{R}^{n}\rightarrow\mathbb{R}^{n\times m}$, $m$ and $n$ are known constants. The dynamical model for an obstacle is given as:
\begin{equation}\label{eq:target_motion}
\mathbf{o}_{k+1}= \bm{\xi}(\mathbf{o}_{k}),
\end{equation}
where $\mathbf{o}_{k}\in\mathbb{R}^{n_{\mathrm{o}}}, n_{\mathrm{o}}\leq n,$ denotes the state of the obstacle at time $k$, and 
$\bm{\xi}(\cdot)$ is a nonlinear state transition function.
\subsection{Discrete-Time Control Barrier Functions}
Let a closed convex set $\mathcal{C}\subset\mathbb{R}^{n}\times\mathbb{R}^{n_{\mathrm{o}}}$ be the $0$-superlevel set of a function $h:\mathbb{R}^{n}\times\mathbb{R}^{n_{\mathrm{o}}}\rightarrow\mathbb{R}$, which is defined as
\begin{equation}\label{Invariant_Set}
		\begin{aligned}
			\mathcal{C}_{k} & \triangleq\left\{(\mathbf{x}_{k},\mathbf{o}_{k})\subseteq \mathbb{R}^{n}\times\mathbb{R}^{n_{\mathrm{o}}}: h(\mathbf{x}_{k},\mathbf{o}_{k}) \geq 0\right\} \\
			\partial \mathcal{C}_{k} & \triangleq\left\{(\mathbf{x}_{k},\mathbf{o}_{k})\subseteq \mathbb{R}^{n}\times\mathbb{R}^{n_{\mathrm{o}}}: h(\mathbf{x}_{k},\mathbf{o}_{k})=0\right\} 
		\end{aligned}
\end{equation}
Herein, we assume that $\mathcal{C}_{k}$ is nonempty.
\begin{Def}\label{CBF_Def}
(Discrete-time CBF~\cite{Discrete_CBF}) Consider the discrete-time system \eqref{Affine_Control_System} and \eqref{eq:target_motion}. Given a set $\mathcal{C}_{k}$ defined by \eqref{Invariant_Set} for a function $h:\mathbb{R}^{n}\times\mathbb{R}^{n_{\mathrm{o}}}\rightarrow\mathbb{R}$, the function $h$ is a CBF defined on set $\mathbb{R}^{n}\times\mathbb{R}^{n_{\mathrm{o}}}$ if there exists a function $\alpha\in\mathcal{K}_{\infty}$ such that
\begin{equation}   \inf\limits_{\mathbf{u}\in\mathcal{U}}\Delta  h(\mathbf{x}_{k},\mathbf{o}_{k},\mathbf{u}_{k})\geq-\alpha(h(\mathbf{x}_{k},\mathbf{o}_{k})), 
\end{equation}
where $\Delta h(\mathbf{x}_{k},\mathbf{o}_{k},\mathbf{u}_{k}):=h(\mathbf{x}_{k+1},\mathbf{o}_{k+1})-h(\mathbf{x}_{k},\mathbf{o}_{k})$.
\end{Def}
We follow the result of \cite{Discrete_CBF} and select the $\mathcal{K}_{\infty}$ function $\alpha(h(\mathbf{x}_{k},\mathbf{o}_{k}))$ to be $\gamma h(\mathbf{x}_{k},\mathbf{o}_{k}), 0<\gamma\leq 1$. Then the condition for CBF is defined as:
\begin{equation}\label{CBC}
\mathrm{CBC}_{k}\triangleq h(\mathbf{x}_{k+1},\mathbf{o}_{k+1})-(1-\gamma)h(\mathbf{x}_{k},\mathbf{o}_{k})\geq 0,
\end{equation}
where $\mathrm{CBC}_{k}=\mathrm{CBC}(\mathbf{x}_{k},\mathbf{o}_{k},\mathbf{u}_{k})$.
\subsection{MPC with CBF}
We consider using MPC with CBF to address the MOCA problem~\eqref{MPC_CBF}. It solves the following constrained finite-time optimization control problem with horizon $N$ at each time instant $t_{k}$ with $i=0, \ldots, N-1$.
\begin{subequations}\label{MPC_CBF}
\begin{align}
&\min _{\mathbf{u}_{k: k+N-1 \mid k}}J(\mathbf{x}_{k:k+N \mid k},\mathbf{u}_{k:k+N-1 \mid k})   \label{Za}\\
&\text { s.t. } \, \mathbf{x}_{k+i+1 \mid k}=\mathbf{f}(\mathbf{x}_{k+i \mid k})+\mathbf{g}(\mathbf{x}_{k+i \mid k})\mathbf{u}_{k+i \mid k},\label{Zb} \\
&\qquad\mathbf{x}_{k+i \mid k} \in \mathcal{X}, \quad\mathbf{u}_{k+i \mid k} \in \mathcal{U},  \label{Zc}\\
&\qquad\mathbf{x}_{k \mid k}=\mathbf{x}_k, \quad\,\,\,\,\mathbf{x}_{k+N \mid k} \in \mathcal{X}_f,\label{Zd}\\
&\qquad\mathrm{CBC}_{k+i|k}\geq 0,\label{Ze}
\end{align}
\end{subequations}
where the cost function in~\eqref{Za} is the sum of the terminal cost $p\left(\mathbf{x}_{k+N \mid k}\right)$ and stage cost $\sum_{i=0}^{N-1} q\left(\mathbf{x}_{k+i \mid k}, \mathbf{u}_{k+i \mid k}\right)$, which means that $J(\mathbf{x}_{k:k+N \mid k},\mathbf{u}_{k:k+N-1 \mid k})=p\left(\mathbf{x}_{k+N \mid k}\right)+\sum_{i=0}^{N-1} q\left(\mathbf{x}_{k+i \mid k}, \mathbf{u}_{k+i \mid k}\right)$; \eqref{Zb} describes the system dynamics; \eqref{Zc} shows the state and input constraints along the horizon; and \eqref{Zd} provides the constraints on initial condition and terminal set $\mathcal{X}_{f}$. The CBF constraint $\mathrm{CBC}_{k+i|k}=h\left(\mathbf{x}_{k+i+1 \mid k}, \mathbf{o}_{k+i+1 \mid k}\right)-(1-\gamma)h\left(\mathbf{x}_{k+i \mid k}, \mathbf{o}_{k+i \mid k}\right)\geq 0$ given in \eqref{Ze} guarantees the forward invariance of the safety set $\mathcal{C}_{k+i|k}$ as defined in \eqref{Invariant_Set}.
\begin{Rmk}
    The constraint $\mathbf{x}_{k+i \mid k} \in \mathcal{X}, i=1, \ldots, N$ (corresponding to the state constraint in~\eqref{Zc} and~\eqref{Zd}) and $\mathrm{CBC}_{k+i|k}\geq 0$ are both state constraints. However, to emphasize the safety constraint of the CBF, we distinguish CBF constraint from the state constraints $\mathbf{x}_{k+i \mid k} \in \mathcal{X}, i=1, \ldots, N$. This distinction highlights the significance of safety as a special constraint. Additionally, there are multiple ways that the MPC-CBF formulation can be infeasible, such as the constraints on $\mathbf{x}_{k+i \mid k} \in \mathcal{X}, i=1, \ldots, N$ and the constraint in~\eqref{Ze} are conflicting.
\end{Rmk}
\subsection{Problem Statement}
Due to inaccurate localization or sensing, perfect measurements of the obstacle are unavailable. We assume that the measurements of the obstacle are corrupted by stochastic noise $\bm{\omega}_{k}$ and are generated with the following model.
\begin{equation}\label{Obstacle_Traj}
\begin{split}
\mathbf{o}_{k+1}=\bm{\xi}(\mathbf{o}_{k})+\bm{\omega}_{k},
\end{split}
\end{equation}
where $\bm{\omega}_{k}$ is white driving noise, which follows a Gaussian distribution $\mathcal{N}(\mathbf{0},\sigma^2\mathbf{I}_{n_{\mathrm{o}}\times n_{\mathrm{o}}})$, and $\sigma$ is the standard deviation of the noise. Since the stochastic noise is introduced in~\eqref{Obstacle_Traj}, the condition in~\eqref{Ze} cannot be satisfied anymore. To address this issue, we are interested in modifying the MPC-CBF in~\eqref{MPC_CBF} to handle stochastic uncertainties. Therefore, the research problem of this paper is formally stated as follows.

\vspace{3mm}
\noindent
\textbf{Problem Statement.}~\textit{Develop a new MPC-CBF formulation to address the MOCA problem in a stochastic environment, which allows designing a controller that is robust to the random noise in~\eqref{Obstacle_Traj}.}

    
\section{Chance Constrained MPC-CBF}\label{Main_results}
In this section, we formulate the CC-MPC-CBF to handle the influence of stochasticity that arises from noisy obstacle measurements. Specifically, we assume that the uncertainty in an obstacle measurement follows a Gaussian distribution, and we transform the chance constraint into a deterministic constraint with their mean and variance. By following the formulation in~\eqref{MPC_CBF}, the CC-MPC-CBF is provided.
\subsection{CBF and Chance Constraints}
Hyperellipsoid is one popular choice of a discrete-time CBF and is commonly used for representing obstacles (or the region of operation where the robot is allowed to move). In our formulation, we continue to use hyperellipsoids to parameterize CBFs, which are denoted as:
\begin{equation}\label{hyperellipsoid}
    \begin{split}
        h(\mathbf{x}_{k},\mathbf{o}_{k})=\|\mathbf{x}_{k}-\mathbf{o}_{k}\|_{\mathbf{W}}^{2}-1,
    \end{split}
\end{equation}
where $\mathbf{W}\in\mathbb{R}^{n_{\mathrm{o}}\times n_{\mathrm{o}}}$ is a symmetric positive definite matrix and $\|\mathbf{x}_{k}-\mathbf{o}_{k}\|_{\mathbf{W}}^2=(\mathbf{x}_{k}-\mathbf{o}_{k})^{\top}\mathbf{W}(\mathbf{x}_{k}-\mathbf{o}_{k})$. It should be noted that in the equation \eqref{hyperellipsoid}, we assumed that the dimension of $\mathbf{x}_{k}$ is $n_{\mathrm{o}}$. However, when the dimension of $\mathbf{x}_{k}$ is greater than $n_{\mathrm{o}}$, the variables in~\eqref{hyperellipsoid} should correspond to a partial state of $\mathbf{x}_{k}$, rather than the full state.

Due to the stochastic noise in~\eqref{Obstacle_Traj}, we consider the chance-constrained optimization problem to accommodate uncertainty with $\delta\in(0,1)$ as the desired confidence of probabilistic safety. Then the chance constraint for collision avoidance is given as follows.
\begin{equation}\label{Standard_Guassian}
    \begin{split}
   \mathbb{P}(\mathrm{CBC}_{k+i|k}\geq\zeta|\mathbf{x}_{k+i|k},\mathbf{o}_{k+i|k},\mathbf{u}_{k+i|k})\geq\delta,
    \end{split}
\end{equation}
where $\mathbb{P}(\cdot)$ denotes the probability of a condition to be true, the value of $\zeta,\delta\in\mathbb{R}^{+}$ are defined by users which vary for different requirements. Herein, $\delta$ indicates the collision avoidance probability.

Note that $ \mathrm{CBC}_{k+i|k}=\|\mathbf{x}_{k+i+1 \mid k}-\bm{\xi}(\mathbf{o}_{k+i|k})\|_{\mathbf{W}}^{2}+\bm{\omega}_{k+i|k}^{\top}\mathbf{W}\bm{\omega}_{k+i|k}+2(\mathbf{x}_{k+i+1 \mid k}-\bm{\xi}(\mathbf{o}_{k+i|k}))^{\top}\mathbf{W}\bm{\omega}_{k+i|k}-(1-\gamma)h\left(\mathbf{x}_{k+i \mid k}, \mathbf{o}_{k+i \mid k}\right)-1$ does not follow a Gaussian distribution. This is due to that there exists a non-Gaussian term $\bm{\omega}_{k+i|k}^{\top}\mathbf{W}\bm{\omega}_{k+i|k}+2(\mathbf{x}_{k+i+1 \mid k}-\bm{\xi}(\mathbf{o}_{k+i|k}))^{\top}\bm{\omega}_{k+i|k}$ in $\mathrm{CBC}_{k+i|k}$. However, as pointed out in~\cite{Approximate_Gaussian}, the non-Gaussian probability density function resulting from the sum of a squaring a Gaussian random variable and a Gaussian term can be well approximated by the Gaussian density function by matching the first-order and second-order moments. 
\begin{Lem}\label{Mean_variance_lemma}
    (\cite{Mean_Covariance_Results}) Let $\mathbf{z}$ be a Gaussian variable, and the mean value and covariance are $\bm{\mu}$ and $\bm{\Sigma}$, respectively. $\mathbf{A}$ is a symmetric matrix. Then the expectation and variance of the quadratic form $\mathbf{z}^{\top}\mathbf{A}\mathbf{z}$ are given as follows.
    \begin{equation}
        \begin{split}
            \mathbb{E}[\mathbf{z}^{\top}\mathbf{A}\mathbf{z}]&=\mathrm{Tr}(\mathbf{A}\bm{\Sigma})+\bm{\mu}^{\top}\mathbf{A}\bm{\mu}, \\ \mathbf{Var}[\mathbf{z}^{\top}\mathbf{A}\mathbf{z}]&= 2 \operatorname{Tr}(\mathbf{A}\bm{\Sigma} \mathbf{A}\bm{\Sigma})+4 \bm{\mu}^{\top} \mathbf{A} \bm{\Sigma} \mathbf{A} \bm{\mu},
        \end{split}
    \end{equation}
where $\mathbb{E}(\cdot)$ and $\mathbf{Var}(\cdot)$ denote the operations for computing the expectation and variance of a variable, respectively.
\end{Lem}
\begin{Cor}
Consider the model~\eqref{Affine_Control_System} and~\eqref{Obstacle_Traj}. We approximate $\mathrm{CBC}_{k+i|k}$ using a Gaussian density, where the parameters, i.e., its expectation and variance, are given as follows. 
\begin{equation}\label{Parameters_CBF_Compact}
    \begin{split} \mathbb{E}\left[\mathrm{CBC}_{k+i|k}\right]&=\mathbf{u}_{k+i|k}^{\top}\bm{\Phi}(\mathbf{x}_{k+i|k})\mathbf{u}_{k+i|k}\\
    &\quad+2\mathbf{m}^{\top}(\mathbf{x}_{k+i|k})\mathbf{u}_{k+i|k}+s(\mathbf{x}_{k+i|k})\\
\mathbf{Var}\left[\mathrm{CBC}_{k+i|k}\right]&=\mathbf{u}_{k+i|k}^{\top}\mathbf{H}(\mathbf{x}_{k+i|k})\mathbf{u}_{k+i|k}\\
&\quad+2\mathbf{n}(\mathbf{x}_{k+i|k})^{\top}\mathbf{u}_{k+i|k}+d(\mathbf{x}_{k+i|k}),
    \end{split}
\end{equation}
where 
\begin{equation*}\label{Parameters_CBF_matrices1}
    \begin{split}  
    \bm{\Phi}(\mathbf{x}_{k+i|k})&=\mathbf{g}(\mathbf{x}_{k+i|k})^{\top}\mathbf{W}\mathbf{g}(\mathbf{x}_{k+i|k})\\
    \mathbf{H}(\mathbf{x}_{k+i|k})&=4\sigma^{2}\mathbf{g}(\mathbf{x}_{k+i|k})^{\top}\mathbf{W}^{\top}\mathbf{W}\mathbf{g}(\mathbf{x}_{k+i|k})\\
    \mathbf{m}(\mathbf{x}_{k+i|k})&=\mathbf{g}(\mathbf{x}_{k+i|k})^{\top}\mathbf{W}(\mathbf{f}(\mathbf{x}_{k+i|k})-\bm{\xi}(\mathbf{o}_{k+i|k}))\\
    \mathbf{n}(\mathbf{x}_{k+i|k})&=4\sigma^{2}\mathbf{g}(\mathbf{x}_{k+i|k})^{\top}\mathbf{W}^{\top}\mathbf{W}(\mathbf{f}(\mathbf{x}_{k+i|k})-\bm{\xi}(\mathbf{o}_{k+i|k}))\\
    s(\mathbf{x}_{k+i|k})&=\|\mathbf{f}(\mathbf{x}_{k+i|k})-\bm{\xi}(\mathbf{o}_{k+i|k})\|_{\mathbf{W}}^{2}+\sigma^{2}\mathrm{tr}(\mathbf{W})\\
    &-(1-\gamma)h\left(\mathbf{x}_{k+i \mid k}, \mathbf{o}_{k+i \mid k}\right)-1\\
    d(\mathbf{x}_{k+i|k})&=4\|\sigma\mathbf{W}(\mathbf{f}(\mathbf{x}_{k+i|k})-\bm{\xi}(\mathbf{o}_{k+i|k}))\|^{2}\\
    &\quad+2\sigma^{4}\mathrm{tr}(\mathbf{W}^{\top}\mathbf{W}).
    \end{split}
\end{equation*}
\end{Cor}
\begin{proof}
Combining the results in~\eqref{CBC}, ~\eqref{Ze}, ~\eqref{Obstacle_Traj}, \eqref{hyperellipsoid}, we have
\begin{equation*}
\begin{split}
\mathrm{CBC}_{k+i|k}=&\|\mathbf{f}(\mathbf{x}_{k+i|k})+\mathbf{g}(\mathbf{x}_{k+i|k})\mathbf{u}_{k+i|k}-\bm{\xi}(\mathbf{o}_{k+i|k})\|_{\mathbf{W}}^{2}\\
&+\bm{\omega}_{k+i|k}^{\top}\mathbf{W}\bm{\omega}_{k+i|k}+2(\mathbf{f}(\mathbf{x}_{k+i|k})\\
&+\mathbf{g}(\mathbf{x}_{k+i|k})\mathbf{u}_{k+i|k}-\bm{\xi}(\mathbf{o}_{k+i|k}))^{\top}\mathbf{W}\bm{\omega}_{k+i|k}\\
&-(1-\gamma)h\left(\mathbf{x}_{k+i \mid k}, \mathbf{o}_{k+i \mid k}\right)-1
\end{split}
\end{equation*}
Next, with the use of Lemma~\ref{Mean_variance_lemma}, we obtain the following result. 
\begin{small}
\begin{equation}\label{Parameters_CBF}
    \begin{split}
    &\mathbb{E}\left[\mathrm{CBC}_{k+i|k}\right]=\|\mathbf{f}(\mathbf{x}_{k+i|k})+\mathbf{g}(\mathbf{x}_{k+i|k})\mathbf{u}_{k+i|k}-\bm{\xi}(\mathbf{o}_{k+i|k})\|_{\mathbf{W}}^{2}\\
    &\quad+\sigma^{2}\mathrm{tr}(\mathbf{W})-(1-\gamma)h\left(\mathbf{x}_{k+i \mid k}, \mathbf{o}_{k+i \mid k}\right)\big)-1\\
   &\mathbf{Var}\left[\mathrm{CBC}_{k+i|k}\right]
   =4\|\sigma\mathbf{W}(\mathbf{f}(\mathbf{x}_{k+i|k})+\mathbf{g}(\mathbf{x}_{k+i|k})\mathbf{u}_{k+i|k}\\
   &\qquad-\bm{\xi}(\mathbf{o}_{k+i|k}))\|^{2}+2\sigma^{4}\mathrm{tr}(\mathbf{W}^{\top}\mathbf{W}).
    \end{split}
\end{equation}
\end{small}
\noindent
Finally, the equalities of ~\eqref{Parameters_CBF_Compact} are obtained by expanding the results regarding $\mathbf{u}_{k+i|k}$ as the variable.
\end{proof}
\subsection{Linear Chance Constraints}

\begin{Lem}
    (\cite{Chance_constrained_lemma}) Given any vector $\mathbf{a}$ and scalar $b$, for a multivariate random variable $\mathbf{z}\in\mathcal{N}(\bm{\mu},\bm{\Sigma})$, then the linear chance constraint 
 \begin{equation}
        \begin{split} \mathbb{P}\left(\mathbf{a}^{\top}\mathbf{z}\leq b\right)\leq\lambda
        \end{split}
    \end{equation}
is equivalent to a deterministic constraint
    \begin{equation}
        \begin{split}
            \mathbf{a}^{\top}\bm{\mu}-b\geq c,
        \end{split}
    \end{equation}  
    where $c=\mathrm{erf}^{-1}(1-2\lambda)\sqrt{2\mathbf{a}^{\top}\bm{\Sigma}\mathbf{a}}$,  $\mathrm{erf}$ is the standard error function and is defined as $\mathrm{erf}(x)=\frac{2}{\sqrt{\pi}}\int\nolimits_{0}^{x}e^{-t^{2}}dt$, and $\lambda>0$ is a use-defined probability.    
\end{Lem}
Note that the error function and its inverse corresponding to the confidence level can be obtained through a look-up table or series approximation techniques.
To ensure that the collision probability is below a certain threshold $1-\delta$, we derive the constraint $\mathbb{P}(\mathrm{CBC}_{k+i|k}<\zeta|\mathbf{x}_{k+i|k},\mathbf{o}_{k+i|k},\mathbf{u}_{k+i|k})< 1-\delta$ from the expression in~\eqref{Standard_Guassian}. By setting $\mathbf{a}=1$, $\mathbf{z}=\mathrm{CBC}_{k+i|k}$, $b=\zeta$, and $\lambda=1-\delta$ in Lemma 1, we obtain the chance constraint~\eqref{Standard_Guassian} into the following form.
\begin{equation}\label{Deterministic_Condition}
    \begin{split}
   \mathbb{E}\left[\mathrm{CBC}_{k+i|k}\right]-\zeta\geq c(\delta)\sqrt{\mathbf{Var}\left[\mathrm{CBC}_{k+i|k}\right]},
    \end{split}
\end{equation}
where $c(\delta)=\sqrt{2}\mathrm{erf}^{-1}(2\delta-1)$. Substituting~\eqref{Parameters_CBF_Compact} into~\eqref{Deterministic_Condition}  gives the following result
\begin{small}
\begin{equation}\label{Re_organize}
    \begin{split}
&c(\delta)\sqrt{\mathbf{u}_{k+i|k}^{\top}\mathbf{H}(\mathbf{x}_{k+i|k})\mathbf{u}_{k+i|k}+2\mathbf{n}(\mathbf{x}_{k+i|k})\mathbf{u}_{k+i|k}+d(\mathbf{x}_{k+i|k})}\\
&-(\mathbf{u}_{k+i|k}^{\top}\bm{\Phi}(\mathbf{x}_{k+i|k})\mathbf{u}_{k+i|k}+2\mathbf{m}(\mathbf{x}_{k+i|k})\mathbf{u}_{k+i|k}+s(\mathbf{x}_{k+i|k}))\\
&\leq -\zeta.
    \end{split}
\end{equation}
\end{small}
\begin{Rmk}\label{CC_Special}
    If the measurement given in~\eqref{Obstacle_Traj} is noise-free (i.e., $\sigma=0$), the variance obtained in~\eqref{Parameters_CBF_Compact} will yield $\mathbb{E}\left[\mathrm{CBC}_{k+i|k}\right]=\|\mathbf{f}(\mathbf{x}_{k+i|k})+\mathbf{g}(\mathbf{x}_{k+i|k})\mathbf{u}_{k+i|k}-\bm{\xi}(\mathbf{o}_{k+i|k})\|_{\mathbf{W}}^{2}-(1-\gamma)h\left(\mathbf{x}_{k+i \mid k}, \mathbf{o}_{k+i \mid k}\right)-1$ and $\mathbf{Var}\left[\mathrm{CBC}_{k+i|k}\right]=0$. Consequently, the inequality constraint in~\eqref{Re_organize} can be reduced to a deterministic constraint, i.e.,~\eqref{Ze}. Furthermore, when the prediction horizon is chosen as $N=1$, the inequality condition~\eqref{Re_organize}, presented below, becomes convex.
    \begin{equation}\label{One_step}
    \begin{split}
&(\mathbf{u}_{k}^{\top}\bm{\Phi}(\mathbf{x}_{k})\mathbf{u}_{k}+2\mathbf{m}(\mathbf{x}_{k})\mathbf{u}_{k}+s(\mathbf{x}_{k}))\\
&\quad-c(\delta)\sqrt{\mathbf{u}_{k}^{\top}\mathbf{H}(\mathbf{x}_{k})\mathbf{u}_{k}+2\mathbf{n}(\mathbf{x}_{k})\mathbf{u}_{k}+d(\mathbf{x}_{k})}\geq \zeta.
    \end{split}
\end{equation}
This is due to that the state-dependent matrices, which include $\bm{\Phi}(\mathbf{x}_{k})$, $\mathbf{H}(\mathbf{x}_{k})$, $\mathbf{m}(\mathbf{x}_{k})$, $\mathbf{n}(\mathbf{x}_{k})$, $d(\mathbf{x}_{k})$, and $s(\mathbf{x}_{k})$, are constant at each time step $k$, given that $\mathbf{x}_{k}=\mathbf{x}_{k|k}$ can be measured or obtained at time $k$. Note that $\bm{\Phi}(\mathbf{x}_{k})$ and $\mathbf{H}(\mathbf{x}_{k})$ are both positive definite if $\mathbf{g}(\mathbf{x}_{k})$ is row full rank and $\mathbf{W}$ is a full rank matrix. 
\end{Rmk}
\subsection{Chance-Constrained MPC-CBF}
For the MPC-CBF formulation given in~\eqref{MPC_CBF}, we usually set the terminal cost to be $p\left(\mathbf{x}_{k+N \mid k}\right)=\|\mathbf{x}_{k+N \mid k}\|_{\mathbf{P}}^{2}$ and the stage cost as $q\left(\mathbf{x}_{k+i \mid k},\mathbf{u}_{k+i \mid k}\right)=\|\mathbf{x}_{k+i \mid k}\|_{\mathbf{Q}}^{2}+\|\mathbf{u}_{k+i \mid k}\|_\mathbf{R}^{2}$, where $\mathbf{P}$, $\mathbf{Q}$, and $\mathbf{R}$ are positive definite weight matrices for the terminal states, stage states, and control inputs. Then the CC-MPC-CBF is formulated as follows.
\begin{equation}\label{MPC_CBF_Problistic}
\begin{split}
&\min_{\mathbf{u}_{k: k+N-1 \mid k}} p\left(\mathbf{x}_{k+N \mid k}\right)+\sum\limits_{i=1}^{N-1}q\left(\mathbf{x}_{k+i \mid k},\mathbf{u}_{k+i \mid k}\right)\\
&\quad\text { s.t. } \eqref{Zb}, \eqref{Zc}, \eqref{Zd}, \eqref{Re_organize}.
\end{split}
\end{equation}
It is important to note that this framework can also be adapted for a reference tracking example with slight modifications. To achieve this, we simply adjust the terminal cost to be $p\left(\mathbf{x}_{k+N \mid k}\right)=\|\mathbf{x}_{k+N \mid k}-\mathbf{x}_{k+N \mid k}^{\mathrm{ref}}\|_{\mathbf{P}}^{2}$, and the stage cost becomes $q\left(\mathbf{x}_{k+i \mid k},\mathbf{u}_{k+i \mid k}\right)=\|\mathbf{x}_{k+i \mid k}-\mathbf{x}_{k+i \mid k}^{\mathrm{ref}}\|_{\mathbf{Q}}^{2}+\|\mathbf{u}_{k+i \mid k}-\mathbf{u}_{k+i \mid k}^{\mathrm{ref}}\|_\mathbf{R}^{2}$, where the superscript ``$\mathrm{ref}$" denotes the reference trajectory. These adjustments enable the framework to handle reference tracking scenarios effectively.

There are several off-the-shelf solvers, such as IPOPT, MOSEK, and CPLEX, that can be used to solve the non-convex optimization problem~\eqref{MPC_CBF_Problistic}. Additionally, it is important to note that the terminal cost can be treated as a control Lyapunov function (CLF) and formulated as a constraint, similar to~\cite{CC-MPC-CLF-CBF}. Then the formulation in~\eqref{MPC_CBF_Problistic} can be viewed as a chance-constrained MPC-CLF-CBF.
\begin{Rmk}\label{Drawbacks}
The major limitation associated with the CC-MPC-CBF approach is feasibility. Define the reachable set $\mathcal{R}_{k+i|k}=\{\mathbf{x}_{k+i|k}\in\mathbb{R}^{n}: \eqref{Zb}, \eqref{Zc},\eqref{Zd}\}, i=0,\cdots, N$ and safety set $\mathcal{S}_{k+i \mid k}=\left\{\mathbf{x}_{k+i \mid k} \in \mathbb{R}^n:~\eqref{Zb}, \mathbf{x}_{k \mid k}=\mathbf{x}_k,(16), \mathbf{u}_{k+i \mid k} \in \mathcal{U}\right\}, i=0, \cdots, N$. The CC-MPC-CBF~\eqref{MPC_CBF_Problistic} requires the intersection set $\mathcal{R}_{k+i|k}\cap\mathcal{S}_{k+i|k}$ is not empty, making the optimization problem possibly infeasible. It is worth highlighting that the motivation behind the definition of $\mathcal{R}_{k+i \mid k}$ and $\mathcal{S}_{k+i \mid k}$ is to achieve sequential implementation in Section IV.
\end{Rmk}
\begin{Prop}\label{Robustness_to_stochasity}
Let $\delta\geq\frac{1+\mathrm{erf}(0.5)}{2}$, and suppose there exists a value of noise $\sigma$ that guarantees the feasibility of~\eqref{MPC_CBF_Problistic} with a probability $\Gamma$. For any real system, if the standard deviation of the noise $\sigma_{\mathrm{upd}}>\sigma$, the CC-MPC-CBF~\eqref{MPC_CBF_Problistic} will encounter feasible cases with a probability $\Gamma_{\mathrm{upd}}\leq\Gamma$.
\end{Prop}
\begin{proof}
The expectation and variance of $\mathrm{CBC}_{k+i|k}$ are given in \eqref{Parameters_CBF}, and $\mathbf{Var}\left[\mathrm{CBC}_{k+i|k}\right]\geq 2\sigma^{4}\mathrm{tr}(\mathbf{W}^{\top}\mathbf{W})$. Meantime, we recall the trace inequality $\sqrt{\mathrm{tr}(\mathbf{W}^{\top}\mathbf{W})}\leq \mathrm{tr}(\mathbf{W})$ with $\mathbf{W}$ being a symmetric positive definite matrix. By substituting~\eqref{Parameters_CBF} and the two inequality conditions into~\eqref{Deterministic_Condition}, it gives
\begin{equation}\label{Prob}
    \begin{split}
    &\mathbb{E}\left[\mathrm{CBC}_{k+i|k}\right]-c(\delta)\sqrt{\mathbf{Var}\left[\mathrm{CBC}_{k+i|k}\right]}-\zeta\\
    &\leq \mathbf{D}+\sigma^{2}\mathrm{tr}(\mathbf{W})-c(\delta)\sqrt{2\sigma^{4}\mathrm{tr}(\mathbf{W}^{\top}\mathbf{W})}-\zeta\\
    &\leq\mathbf{D}-\zeta+\left(1-\sqrt{2}c(\delta)\right)\sqrt{\mathrm{tr}(\mathbf{W}^{\top}\mathbf{W})}\sigma^{2},
    \end{split}
\end{equation}
where $\mathbf{D}=\|\mathbf{f}(\mathbf{x}_{k+i|k})+\mathbf{g}(\mathbf{x}_{k+i|k})\mathbf{u}_{k+i|k}-\bm{\xi}(\mathbf{o}_{k+i|k})\|_{\mathbf{W}}^{2}-(1-\gamma)h\left(\mathbf{x}_{k+i \mid k}, \mathbf{o}_{k+i \mid k}\right)\big)-1$ is deterministic if $\mathbf{x}_{k+i|k}$, $\mathbf{o}_{k+i|k}$, and $\mathbf{u}_{k+i|k}$ are given. Next, given that $\delta\geq\frac{1+\mathrm{erf}(0.5)}{2}$, we can infer that $1-\sqrt{2}c(\delta)\leq 0$. Meantime, according to the assumption that $\sigma$ is the value able to guarantee the satisfaction of the constraints~\eqref{MPC_CBF_Problistic} with a probability $\Gamma$, then for the case that $\sigma_{\mathrm{upd}}>\sigma$, it will lead to a decrease in the safety margin defined by the inequality~\eqref{Deterministic_Condition}. Consequently, the size of the safety set $\mathcal{S}_{k+i|k}$ decreases, which leads to a reduction in the size of the feasible set $\mathcal{R}_{k+i|k}\cap\mathcal{S}_{k+i|k}$. Therefore, the probability of the constrained MPC in \eqref{MPC_CBF_Problistic} being feasible also decreases, which infers $\Gamma_{\mathrm{upd}}\leq\Gamma$.
\end{proof}
Proposition~\ref{Robustness_to_stochasity} highlights that the feasibility of~\eqref{MPC_CBF_Problistic} depends on the specified collision avoidance probability $\delta$. This is because the trade-off between robustness to stochasticity and feasibility cannot be simultaneously achieved in~\eqref{MPC_CBF_Problistic}. For instance, when the collision avoidance probability $\delta$ is set to be high and the noise is highly stochastic, a conservative control strategy may be required to ensure safety, which may render~\eqref{MPC_CBF_Problistic} infeasible. Conversely, when a small value of $\delta$ is chosen, allowing for some level of failure in collision avoidance, the optimization problem~\eqref{MPC_CBF_Problistic} remains feasible even in the presence of high noise variability.   
\begin{Rmk}\label{Hyperparameter_Gamma}
    The feasibility of \eqref{Re_organize} is influenced by the parameter $\gamma$, which has the same impact in both deterministic~\cite{MPC_CBF} and non-deterministic scenarios. As $\gamma$ is decreased, the parameter $s(\mathbf{x}_{k+i|k})$ defined in~\eqref{Parameters_CBF_Compact} is reduced, leading to a corresponding decrease in the expectation $\mathbb{E}\left[\mathrm{CBC}_{k+i|k}\right]$ associated with $s(\mathbf{x}_{k+i|k})$. We  rewrite \eqref{Deterministic_Condition} as $c(\delta)\sqrt{\mathbf{Var}\left[\mathrm{CBC}_{k+i|k}\right]}+\zeta \leq \mathbb{E}\left[\mathrm{CBC}_{k+i|k}\right]$, which shows that the upper bound of the inequality, i.e., $\mathbb{E}\left[\mathrm{CBC}_{k+i|k}\right]$, is decreased. Consequently, as $\mathbb{E}\left[\mathrm{CBC}_{k+i|k}\right]$ decreases, the safety set $\mathcal{S}_{k+i|k}$ becomes smaller, resulting in a smaller feasible set $\mathcal{R}_{k+i|k} \cap \mathcal{S}_{k+i|k}$. Therefore, the probability of \eqref{MPC_CBF_Problistic} being feasible will be decreased with a decrease of $\gamma$. Moreover, when $\gamma=1$, the chance-constrained CBF \eqref{Standard_Guassian} reduces to the form $\mathbb{P}(h(\mathbf{x}_{k+i|k},\mathbf{o}_{k+i|k})\geq\zeta|\mathbf{x}_{k+i|k},\mathbf{o}_{k+i|k},\mathbf{u}_{k+i|k})\geq\delta$, which is recognized as a CC-MPC with distance constraints (CC-MPC-DC)~\cite{Chance_Constraied}.
\end{Rmk}
\begin{Rmk}\label{Hyper_N}
    Increasing the prediction horizon will raise the probability of infeasibility. This is because the introduction of more constraints, resulting from an increase in the value of $N$, will cause both sets $\mathcal{R}_{k+i|k}, i=0,\cdots, N$ and $\mathcal{S}_{k+i|k}$ to contract. It may result in an empty set of $\mathcal{R}_{k+i|k} \cap \mathcal{S}_{k+i|k}$.  
\end{Rmk}

It is worthy to emphasize that the discussion regarding the impact of $\gamma$ and $N$ is specified at a particular time intervals $\{k,\cdots,k+N-1\}$ rather than the entire task duration. When considering a specific task over the entire time interval, one might argue that increasing $\gamma$ or $N$ can result in earlier actions (smaller control inputs), potentially reducing the possibility of encountering infeasibility.
\section{CC-MPC-CBF with a Sequential Implementation}\label{Sequantial_Approach}
In this section, we present a solution to address the limitation outlined in Remark~\ref{Drawbacks}.  A common strategy in practice is to relax the general state constraints, typically achieved by introducing slack variables to the constraints, as seen in related works~\cite{MPC_Relaxation}.  Specifically, we propose a sequential implementation of the CC-MPC-CBF approach. This involves decomposing \eqref{MPC_CBF_Problistic} into two sub-optimization problems. The first sub-optimization problem utilizes the MPC formulation without CBF constraints, resulting in a standard MPC formulation as follows, where its objective is to provide a nominal control input $\mathbf{u}_{k: k+N-1 \mid k}^{\mathrm{nom}}$ that guarantees system performance, such as closed-loop stability.
\begin{equation}\label{MPC_Standard}
\begin{aligned}
J_k^*\left(\mathbf{x}_k\right)= & \min _{\mathbf{u}_{k k+N-1 \mid k}} J\left(\mathbf{x}_{k \cdot k+N \mid k}, \mathbf{u}_{k \cdot k+N-1 \mid k}\right) \\
& \text { s.t.~\eqref{Zb},~\eqref{Zc},~\eqref{Zd}. }
\end{aligned}
\end{equation}
In the second sub-optimization problem, we ensure safety by incorporating the CBF constraints. It is formulated as follows, and we call it a predictive safety filter.
\begin{equation}\label{CBF}
\begin{split}
&\min_{\mathbf{u}_{k: k+N-1 \mid k}} \|\mathbf{u}_{k\mid k}-\mathbf{u}_{k\mid k}^{\mathrm{nom}}\|^{2}\\
&\quad\text { s.t. } \eqref{Zb}, \mathbf{x}_{k \mid k}=\mathbf{x}_k, \eqref{Re_organize}, \\
&\qquad\quad \mathbf{u}_{k+i|k}\in\mathcal{U}
\end{split}
\end{equation}
We use the off-the-shelf solvers, such as IPOPT, MOSEK, and CPLEX, to solve the above two optimization problems.

The first sub-optimization problem~\eqref{MPC_Standard} has been well-studied in many existing pieces of literature~\cite{MPC_Literature1,MPC_Literature2}, which can be used for stabilization and reference trajectory tracking applications. For the second sub-optimization problem, it is based on the idea of the safety filter, or an active set invariance filter~\cite{Safety_filter}, where the nominal control input is filtered through~\eqref{CBF} with safety guarantees because of~\eqref{Re_organize}. Additionally, it is noteworthy to mention that splitting~\eqref{MPC_CBF_Problistic} into a standard MPC and~\eqref{CBF} could be a natural option for numerous applications. In real-world systems, a standard MPC is formulated without considering safety constraints, making it a preferred scenario to introduce~\eqref{CBF}.
\begin{Rmk}\label{CC_MPC_CBF_Feasibility_Analysis}
The feasible sets for the standard MPC~\eqref{MPC_Standard} and the predictive safety filter~\eqref{CBF} are denoted by $\mathcal{R}_{k+i|k}$ and $\mathcal{S}_{k+i|k}$, respectively. Notably, the feasible set for both sub-optimization problems is larger when compared to the intersection $\mathcal{R}_{k+i|k} \cap \mathcal{S}_{k+i|k}$ of the CC-MPC-CBF~\eqref{MPC_CBF_Problistic}, thereby reducing the possibility of encountering infeasible cases. However, it should be emphasized that, in this formulation, the general state constraint $\mathbf{x}_{k+i \mid k} \in \mathcal{X}$, for $i=1, \ldots, N$, is no longer strictly regarded as a hard constraint, meaning it may not be strictly satisfied. The predictive safety filter modifies the nominal control input generated by MPC, which could potentially sacrifice some degree of optimality to ensure the satisfaction of $\mathbf{x}_{k+i \mid k} \in \mathcal{X}$. Furthermore, while the sequential implementation increased feasibility, it is important to note that the two separate optimization problems can still become infeasible in certain scenarios. To address this feasibility issue, related solutions can be explored in prior works~\cite{MPC_Feasibility} and~\cite{predictive_safe_filter}.
\end{Rmk}
\begin{Rmk}
    In the special case that $N=1$ in~\eqref{CBF}, it is corresponding to the well-known safety filter, as detailed in~\cite{Safety_filter}. However, despite its widespread recognition, this filter comes with a drawback: it sacrifices its ability to make predictions and may cause excessively aggressive behavior. 
\end{Rmk}

\subsection{Iterative Convex Optimization}~\label{Convexify_Result}
As mentioned, the first sub-optimization problem has already been extensively studied, and hence we will not provide any further descriptions in this paper. However, it is interesting to have some discussions on the second sub-optimization problem~\eqref{CBF}.

We notice that the optimization problem~\eqref{CBF} is convex if a specific set of $\mathbf{x}_{k:k+N|k}$ is previously given, where the convexity of the chance-constrained constraint has already been discussed in Remark~\ref{CC_Special}.
Therefore, we suggest using an iterative convex optimization algorithm, as outlined in Algorithm~\ref{alg:Framwork}, to solve ~\eqref{CBF}. The proposed algorithm aims to solve~\eqref{CBF} using a sequence of convex optimization problems in an iterative manner. In particular, the system state $\mathbf{x}_{k:k+N|k}$ of \eqref{CBF} is fixed and updated at each iteration while the control input is treated as the decision variable. By applying a suitable stop criterion, we can obtain a near-optimal solution of \eqref{CBF}. Note that the variable $j$ represents the iteration number, which ranges from 1 to $j_{\mathrm{max}}$, where $j_{\mathrm{max}}$ is the maximum number of allowed iterations. The overall duration of the time interval under consideration is denoted by $k_{\mathrm{max}}$. The stop criterion at each $k$ is given as follows.
\begin{equation}\label{stop_criterion} 
\sum\limits_{i=0}^{N}\|\mathbf{x}_{k+i|k}^{j+1}-\mathbf{x}_{k+i|k}^{j}\|\leq \varepsilon, 
\end{equation}
where $\varepsilon>0$ is constant, which can be obtained after some offline trials. It is worth noting that a weighting matrix can be included in the criterion~\eqref{stop_criterion}. However, for simplicity, we have not taken it into account in this case. Additionally, it should be noted that we make the assumption that the starting state $\mathbf{x}_{0}$ does not encounter any obstacles and that the robot has the capability to reach the state $\mathbf{x}_{k+1|k}^{*}$ (obtained in Step 10 of Algorithm~\ref{alg:Framwork}) before $k+1$. 
\begin{algorithm}[tb]
\caption{Iterative Convex Optimization}
\label{alg:Framwork}
\begin{algorithmic}[1]
\REQUIRE
The system model~\eqref{Affine_Control_System} and obstacle model~\eqref{Obstacle_Traj}, 

$\quad\,\,$ state $\mathbf{x}_{k}$ at each $k, k=0,\cdots, k_{\mathrm{max}}$.
\ENSURE Closed-loop trajectory $\mathbf{x}_{k:k+N|k}^{*}$ at each $k$.
\STATE  Set $k=0$ and initial guess $\mathbf{u}_{0:N-1|0}^{0}=\mathbf{u}_{k: k+N-1 \mid k}^{\mathrm{nom}}$.
\STATE Initialize the variables $\mathbf{x}_{0:N|0}^{0}$ using $\mathbf{x}_{0}$ and $\mathbf{u}_{0:N-1|0}^{0}$ via~\eqref{Zb}.
\STATE \textbf{for} $k<k_{\mathrm{max}}$ \textbf{do}
\STATE $\quad$ Set $j=0$.
\STATE $\quad$ \textbf{While} $j<j_{\mathrm{max}}$ or \eqref{stop_criterion} is not satisfied \textbf{do}
\STATE \qquad \,\,\, Fix $\mathbf{x}_{k:k+N|k}^{j}$ and compute the optimal solution 

\qquad  \,\,\, $\mathbf{u}_{k:k+N-1|k}^{j+1}$ in~\eqref{CBF} through some convex solvers 

\qquad \,\,\, like CVXOPT and GUROBI.

\STATE \qquad \,\,\, \textbf{Update} the variables $\mathbf{x}_{k:k+N|k}^{j+1}$ for the next 

\qquad \,\,\, iteration using $\mathbf{x}_{k}$ and $\mathbf{u}_{k:k+N-1|k}^{j+1}$ via~\eqref{Zb}.
\STATE \qquad \,\,\, $j=j+1$.
\STATE $\quad$ \textbf{end while}
\STATE Extract optimized states $\mathbf{x}_{k:k+N|k}^{*}=\mathbf{x}_{k:k+N|k}^{j}$ and inputs $\mathbf{u}_{k:k+N-1|k}^{*}=\mathbf{u}_{k:k+N-1|k}^{j}$ from the last iteration.
\STATE Update  $\mathbf{u}_{k+1:k+N|k+1}^{0}$ with $\mathbf{u}_{k:k+N-1|k}^{*}$ and obtain the states $\mathbf{x}_{k+1:k+N+1|k+1}^{0}$ for the next time instant with $\mathbf{x}_{k}$ and $\mathbf{u}_{k:k+N-1|k}^{*}$ according to~\eqref{Zb}.
\STATE $k=k+1$.
\STATE \textbf{Return} closed-loop trajectory $\mathbf{x}_{k:k+N|k}^{*}$. 
\STATE \textbf{end for}
\end{algorithmic}
\end{algorithm}

In order to provide a clear explanation of the iterative convex optimization algorithm, we will tour the readers to check some steps of   Algorithm~\ref{alg:Framwork} with detailed explanations. At Step 2 or Step 7 for each iteration, we use $\mathbf{x}_{k}, k=0,\cdots,{k}_{\mathrm{max}}$ and $\mathbf{u}_{k:k+N-1|k}^{0}$ to obtain $\mathbf{x}_{k:k+N|k}^{0}$. If $\mathbf{x}_{k+i|k}^{0}\in\mathcal{S}_{k+i|k}, i=0,\cdots, N$, then $\mathbf{u}_{k:k+N-1|k}^{0}$ would be the solution of~\eqref{CBF}, as the safety requirement is already satisfied with the solution. However, if $\mathbf{x}_{k+i|k}^{0}\notin\mathcal{S}_{k+i|k}$, which means that~\eqref{CBF} cannot be satisfied with $\mathbf{x}_{k:k+N|k}^{0}$ and $\mathbf{u}_{k:k+N-1|k}^{0}$, this leads to Step 6. In Step 6, we obtain $\mathbf{u}_{k:k+N-1|k}^{1}$, which ensures that $\mathbf{x}_{k:k+N|k}^{0}$ with $\mathbf{u}_{k:k+N-1|k}^{1}$ satisfies~\eqref{CBF}. We repeat the above processes until we find a $\mathbf{x}_{k:k+N|k}^{j}\in\mathcal{S}_{k+i|k}$ or until $j=j_{\mathrm{max}}$.
\begin{Rmk}
Note that the stop criterion~\eqref{stop_criterion} in Algorithm~\ref{alg:Framwork} guarantees a near-optimal solution for~~\eqref{CBF}. This is because at each iteration $j$, solving~\eqref{CBF} ensures that the predicted state $\mathbf{x}_{k:k+N|k}^{j}$ with control input $\mathbf{u}_{k:k+N|k}^{j+1}$ satisfies~\eqref{CBF} (as shown in Step 6 of Algorithm~\ref{alg:Framwork}). The stop criterion~\eqref{stop_criterion} ensures that the difference between the predicted states $\mathbf{x}_{k:k+N|k}^{j+1}$ and $\mathbf{x}_{k:k+N|k}^{j}$ is within a certain tolerance level. Therefore, we claim that $\mathbf{x}_{k:k+N|k}^{j+1}$ is a near-optimal solution to~\eqref{CBF}.
\end{Rmk}
\section{Application Examples and Simulation Results}\label{Simulation}
This section applies the CC-MPC-CBF approach developed in Section~\ref{Main_results} and the sequential implementation approach presented in Section~\ref{Sequantial_Approach} to a double integrator system for MOCA. Through simulations, we highlight the benefits of the proposed approach in terms of its robustness to stochastic noises, fast computation speed for addressing the CC-MPC-CBF optimization problem, and desired feasibility.
\subsection{Dynamical Model of Robot}
We consider a linear discrete-time system modelled by a double integrator 
\begin{equation}
\mathbf{x}_{k+1}=\mathbf{A}\mathbf{x}_{k}+\mathbf{B}\mathbf{u}_{k},
\end{equation}
where $\mathbf{x}_{k}=[\mathbf{p}_{k}, \mathbf{v}_{k}]^{\top}$ is the system state, and $\mathbf{p}_{k}$ and $\mathbf{v}_{k}$ denote position and velocity of the robot, respectively. The control input $\mathbf{u}_{k}$ is the acceleration. The system matrices are given by 
\begin{equation*}
    \mathbf{A}=\left[\begin{array}{ll}
\mathbf{1} & \Delta T \\
\mathbf{0} & \mathbf{0}
\end{array}\right]\quad \mathbf{B}=\left[\begin{array}{l}
\mathbf{0} \\
\mathbf{I}
\end{array}\right].
\end{equation*}
Note that $\mathbf{f}(\mathbf{x})=\mathbf{A}\mathbf{x}_{k}$ and $\mathbf{g}(\mathbf{x})=\mathbf{B}$ in this case.
\subsection{Obstacle Model}
We model each obstacle as a non-rotating enclosing spherical. The motion of the obstacle follows:
\begin{equation}
\mathbf{o}_{k+1}=\mathbf{A}\mathbf{o}_{k}+\mathbf{B}\bm{\tau}_{k}+\bm{\omega}_{\mathbf{o}_{k}},
\end{equation}
where the state of the obstacle at time $k$ is characterized by both its position $\mathbf{o}_{\mathbf{p},{k}}$ and velocity $\mathbf{o}_{\mathbf{v},{k}}$, which is denoted as $\mathbf{o}_{k}=[\mathbf{o}_{\mathbf{p},{k}},\mathbf{o}_{\mathbf{v},{k}}]$. $\bm{\tau}_{k}$ is the control input of the obstacle at time $k$.  We assume the  obstacle position at each $k$ is available, while the velocity measurements $\bm{\omega}_{\mathbf{o}_{\mathbf{v}_{k}}}$ are corrupted by a Gaussian noise which follows $\mathcal{N}(\mathbf{0},\sigma^{2}\mathbf{I}_{3\times 3})$.
\subsection{Configurations}\label{Sim_Conf}
\subsubsection{Reference Trajectory and Obstacle Settings}
We consider a desired reference trajectory $\mathbf{r}_{d}=[2\sin(0.4 k), 2\cos(0.4 k), 2]^{\top}$, which is a circle in $\mathrm{X-Y}$ plane centered at the origin. We require the robot, initialized at the position $\mathbf{p}_{0}=[0,0,2]^{\top}$,  to follow $\mathbf{r}_{d}$ and avoid two sphere-shaped obstacles of radius $r_{1}=r_{2}=0.8$ with the angular velocities $\Omega_{1}=0.8 \mathrm{rad/s}$ and $\Omega_{2}=0.4 \mathrm{rad/s}$, respectively. The total simulation time for the trajectory is set to be $t_{\mathrm{total}}=20\mathrm{s}$. The sampling time is $\Delta T =0.1$, and hence $k_{\mathrm{max}}=200$. For the measurement noise of the obstacle, we set $\sigma^{2}=0.1$.
\begin{figure}[tp]
 \centering
    \makebox[0pt]{%
    \includegraphics[width=3.7in]{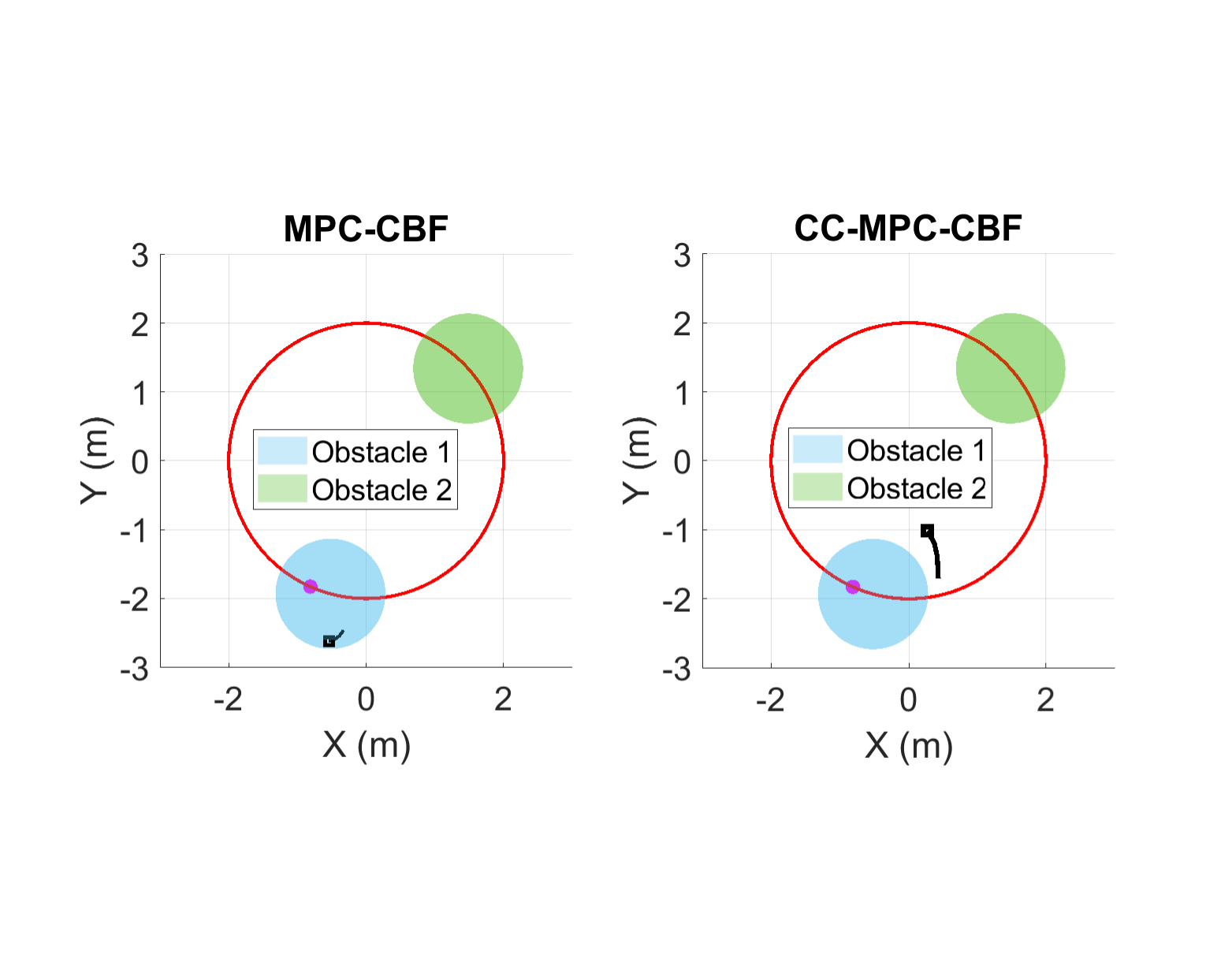}}
    \caption{A snapshot of the tracking behaviors of the deterministic MPC-CBF and CC-MPC-CBF at time $k=90$ (top-down view): The red circle denotes the reference trajectory; the magenta solid point and black box maker represent the reference point and tracking point, respectively. Note that the black box is followed by a tail, which relates to the history tracking data.}
    \label{Comparison_Fail_succ}
\end{figure}
\begin{figure}[tp]
 \centering
    \makebox[0pt]{%
    \includegraphics[width=3.5in]{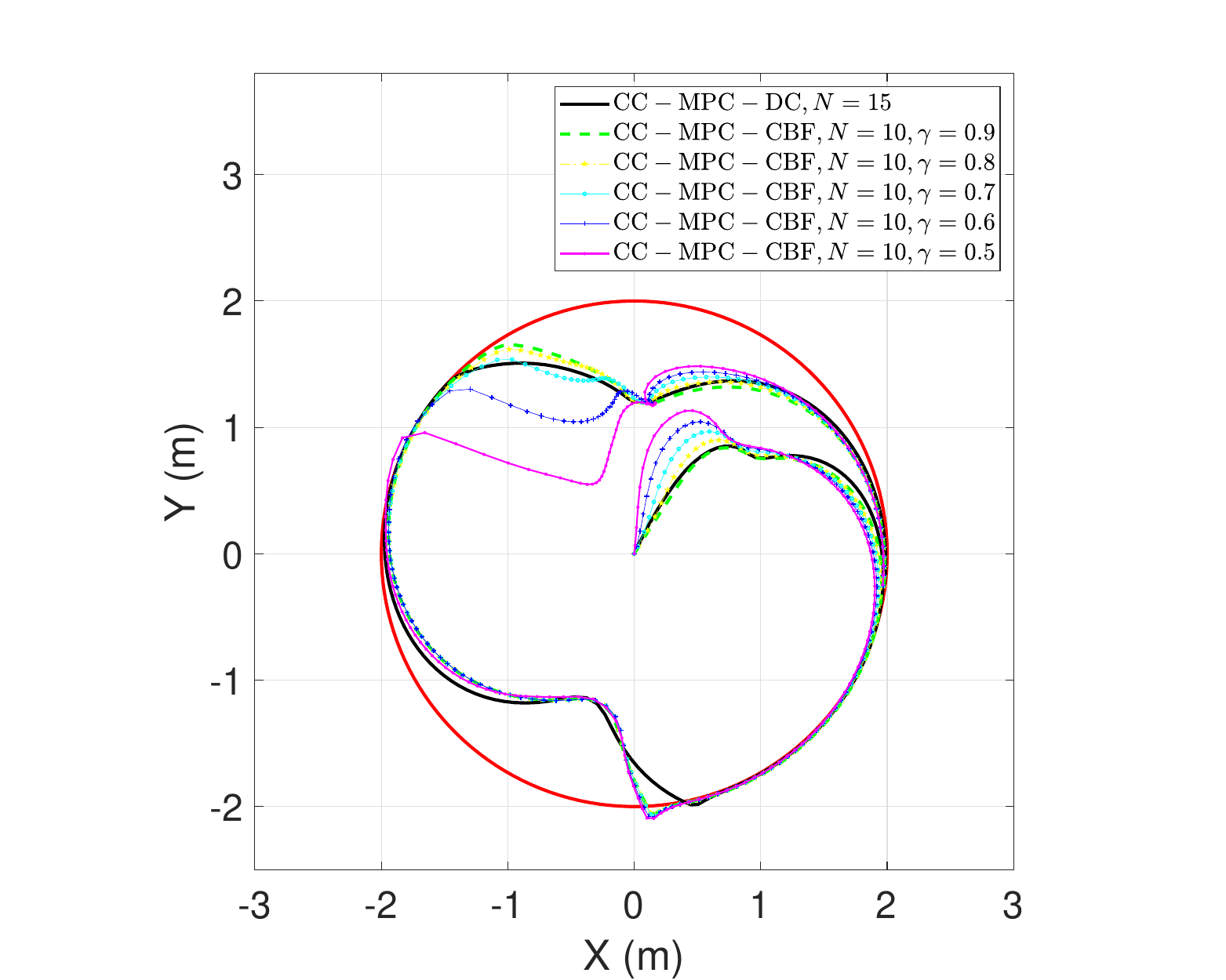}}
    \caption{A comparison of the tracking performance of CC-MPC-DC and the proposed CC-MPC-CBF (with different $N$ and $\gamma$) (top-down view): The red circle denotes the reference trajectory; The curves with different line styles and colors indicate the tracking trajectories with different parameters.}
    \label{Diff_Gamma_Behaviour}
\end{figure}
\subsubsection{CC-MPC-CBF Settings}
For the CC-MPC-CBF optimization problem formulated in Section~\ref{Main_results}, we set the weighting matrices to be $\mathbf{P}=\mathbf{Q}=1000\mathbf{I}_{6\times 6}$ and $\mathbf{R}=\mathbf{I}_{3\times 3}$. The prediction horizon $N=15$. The CBF is defined by:
\begin{equation}\label{CBF_Real}
    \begin{split}
        h(\mathbf{p}_{k},\mathbf{o}_{k})=\|\mathbf{p}_{k}-\mathbf{o}_{\mathbf{p},k}\|_{\mathbf{W}}^{2}-1.
    \end{split}
\end{equation}
The weighting matrix $\mathbf{W}$ for obstacle 1 and obstacle 2 are set to be $\mathbf{W}_{\mathbf{o},1}=\mathrm{diag}([\frac{1}{r_{1}^2}\mathbf{I}_{1\times 3}])$ and $\mathbf{W}_{\mathbf{o},2}=\mathrm{diag}([\frac{1}{r_{2}^2}\mathbf{I}_{1\times 3}])$, respectively. The system is subject to state constraint $\mathbf{x}_{k}\in\mathcal{X}$ and input constraint $\mathbf{u}_{k}\in\mathcal{U}$,
\begin{equation}
\begin{aligned}
& \mathcal{X}=\left\{\mathbf{x}_k \in \mathbb{R}^n: \mathbf{x}_{\min } \leq \mathbf{x}_k \leq \mathbf{x}_{\max }\right\}, \\
& \mathcal{U}=\left\{\mathbf{u}_k \in \mathbb{R}^m: \mathbf{u}_{\min } \leq \mathbf{u}_k \leq \mathbf{u}_{\max }\right\}.
\end{aligned}
\end{equation}
The lower and upper bounds are
\begin{equation}
\mathbf{x}_{\max }, \mathbf{x}_{\min }= \pm 5 \cdot \mathbf{I}_{6 \times 1}, \mathbf{u}_{\max }, \mathbf{u}_{\min }= \pm 4\cdot  \mathbf{I}_{3 \times 1}.
\end{equation}
Moreover, we define the collision avoidance probability threshold to be $\delta=0.03$ (thus confidence level $0.97$), which corresponds to $3\sigma$ confidence ellipsoid.
\subsubsection{Computer and Solver Settings}
The simulated data is processed using Matlab 2019a on a 64-bit Intel core i7-9750H with a 2.6-GHz processor. Both the optimization problems formulated by CC-MPC-CBF, which is given by~\eqref{MPC_CBF_Problistic}, and the standard MPC problems mentioned in Section~\ref{Sequantial_Approach}, have been solved using a non-convex solver called IPOPT, which is implemented in the CaSadi framework that employs the YAMIP modeling language.
\subsection{Performances of Different Approaches}




\subsubsection{Tracking Performance of the Deterministic MPC-CBF and CC-MPC-CBF}
We compare the tracking performance of deterministic MPC-CBF and CC-MPC-CBF in a stochastic setting. Fig.~\ref{Comparison_Fail_succ} shows a snapshot of the tracking behavior at $k=90$. The actual trajectory (denoted by the black box) fails to avoid the moving obstacles when using deterministic MPC-CBF to track the desired trajectory $\mathbf{r}_{d}$. In contrast, CC-MPC-CBF successfully avoids collisions and achieves satisfying tracking performance. The successful collision avoidance rates with different levels of noise for both methods are provided in TABLE~\ref{Noise_Comp_1} and TABLE~\ref{Noise_Comp_2}. As shown in TABLE~\ref{Noise_Comp_1} and TABLE~\ref{Noise_Comp_2}, the performance of the deterministic MPC-CBF in ensuring successful MOCA is evaluated under different levels of noise, with $\sigma^{2}$ ranging from $0.0001$ to $0.5$ ($100$ trials conducted for each case). The results indicate that the deterministic MPC-CBF approach is not robust enough to guarantee successful MOCA in the presence of noise. As the noise level increases, the successful MOCA rate decreases. In contrast, the CC-MPC-CBF is robust to large sensing uncertainties as the successful collision avoidance rate remains $100\%$ even when $\sigma^{2}$ goes to $0.6$. 
\begin{table*}[tp]
\centering
\caption{Successful collision avoidance rate versus noise variance $\sigma^{2}$ for deterministic MPC-CBF and CC-MPC-CBF.}
\begin{tabular}{ccccc}
  \hline
  Noise & $\sigma^{2}=0$ & $\sigma^{2}=0.0001$ & $\sigma^{2}=0.005$ & $\sigma^{2}=0.01$ \\
  \hline
  Deterministic MPC-CBF (\%) & 100 & 91 & 53 & 38 \\
  CC-MPC-CBF (\%) & 100 & 100 & 100 & 100 \\
  \hline
\end{tabular}
\label{Noise_Comp_1}
\end{table*}
\begin{table*}[tp]
\centering
\caption{(Continued) Successful avoidance rate versus noise variance $\sigma^{2}$ for deterministic MPC-CBF and CC-MPC-CBF.}
\begin{tabular}{ccccc}
  \hline
  Noise & $\sigma^{2}=0.1$ & $\sigma^{2}=0.3$ & $\sigma^{2}=0.5$ & $\sigma^{2}=0.6$ \\
  \hline
  Deterministic MPC-CBF (\%) & 17 & 12 & 8 & 7 \\
  CC-MPC-CBF (\%) & 100 & 100 & 100 & 100 \\
  \hline
\end{tabular}
\label{Noise_Comp_2}
\end{table*}
\subsubsection{Feasibility of CC-MPC-CBF with Different Parameters}
To verify the observations presented in Proposition~\ref{Robustness_to_stochasity}, we conduct $100$ trials with varying levels of noise and set $\delta = 0.97 > \frac{1+\mathrm{erf}(0.5)}{2}$ to evaluate the feasibility of CC-MPC-CBF. The results are summarized in TABLE~\ref{Noise_feasi_1} and TABLE~\ref{Noise_feasi_2}, which indicates that the feasibility of~\eqref{MPC_CBF_Problistic} decreases rapidly with an increase in $\sigma^2$. Furthermore, we recorded the instances of infeasibility, i.e., the time step $k$ at which it occurred, with different $\sigma^2$. The average value of $k$ across the 100 trials is provided in TABLE~\ref{Noise_feasi_1} and TABLE~\ref{Noise_feasi_2}, which reveals that this value decreases with increasing $\sigma^2$. This finding further supports the claim made in Proposition~\ref{Robustness_to_stochasity}. 

The results presented in TABLE~\ref{gamma_1}  and TABLE~\ref{gamma_2} provide compelling support for the claim made in Remark~\ref{Hyperparameter_Gamma}. We initially set $\delta=0.97$ and $\sigma^{2}=1$. As indicated in TABLE~\ref{gamma_1} and TABLE~\ref{gamma_2}, some scenarios may be infeasible under the setting that $\gamma=0.5$. Next, we test the feasibility of CC-MPC-CBF with different $\gamma$. As shown in TABLE~\ref{gamma_1}  and TABLE~\ref{gamma_2}, the results show that decreasing $\gamma$ leads to a higher rate of infeasibility. As mentioned in Remark~\ref{Hyper_N}, increasing the prediction horizon $N$ may lead to infeasibility. To demonstrate this, we again set $\delta=0.97$, $\sigma^{2}=1$, and $\gamma=0.5$ and observe in TABLE~\ref{Prediction_Hor} that a larger $N$ results in more significant feasibility issues.
\subsubsection{Tracking Performance of the CC-MPC-DC and CC-MPC-CBF}
In Fig.~\ref{Diff_Gamma_Behaviour}, we compare the tracking performance of CC-MPC-DC~\cite{Chance_Constraied} and the proposed CC-MPC-CBF (with different $N$ and $\gamma$). Our results demonstrate that CC-MPC-CBF maintains the advantages of early obstacle avoidance when compared to CC-MPC-DC. Specifically, when we set $N=10$ and $\gamma=0.7$, CC-MPC-CBF starts avoiding obstacles at almost the same time as CC-MPC-DC with $N=15$. This means that we can achieve earlier obstacle avoidance without increasing the prediction horizon. Moreover, Fig.~\ref{Diff_Gamma_Behaviour} shows that an increase in $\gamma$ results in more conservative behavior.
\subsubsection{Feasibility Comparison: CC-MPC-CBF vs Sequential Implementation} 
By comparing the results shown in TABLE~\ref{Noise_feasi_1} and TABLE~\ref{Noise_feasi_2}, we can observe that the sequential implementation approach improves the feasibility of the CC-MPC-CBF formulation significantly, as stated in Remark~\ref{CC_MPC_CBF_Feasibility_Analysis}. 
\begin{table*}[t]
\footnotesize
\centering
\caption{Feasibility of CC-MPC-CBF with versus noise variance $\sigma^{2}$.}
\begin{tabular}{cccccccc}
  \hline
  Noise & $\sigma^{2}=0.7$ & $\sigma^{2}=0.8$ &$\sigma^{2}=0.9$ & $\sigma^{2}=1$ & $\sigma^{2}=2$ \\
  \hline
  CC-MPC-CBF (\%) & 100 & 100 & 100 & 93 & 60 \\
  Infeasible at time $k$ & \textbackslash & \textbackslash & \textbackslash & $k=191$ & $k=123$ \\
  Sequential Implementation (\%) & 100 & 100 & 100 & 100 & 100 \\
  Infeasible at time $k$ & \textbackslash & \textbackslash & \textbackslash & \textbackslash & \textbackslash \\
  \hline
\end{tabular}
\label{Noise_feasi_1}
\end{table*}
\begin{table*}[tp]
\footnotesize
\centering
\caption{(Continued) Feasibility of CC-MPC-CBF versus noise variance $\sigma^{2}$.}
\begin{tabular}{cccccccc}
  \hline
  Noise & $\sigma^{2}=3$ & $\sigma^{2}=4$ & $\sigma^{2}=5$ & $\sigma^{2}=6$ \\
  \hline
  CC-MPC-CBF (\%) & 35 & 10 & 2 & 0 \\
  Infeasible at time $k$ & $k=71$ & $k=48$ & $k=33$ & $k=23$ \\
  Sequential Implementation (\%) & 89 & 81 & 73 & 67 \\
  Infeasible at time $k$ & $k=182$ & $k=167$ & $k=144$ & $k=122$ \\
  \hline
\end{tabular}
\label{Noise_feasi_2}
\end{table*}
\subsubsection{Computation efficiency of the sequential implementation approach}
When the settings described in Section~\ref{Sim_Conf} are applied, it is important to note that the sequential approach proves to be more time-efficient compared to solving the CC-MPC-CBF optimization problem in a single run. In our simulation, after performing 100 trials, the average execution time of the sequential implementation amounts to $21.3\mathrm{s}$, which is considerably faster than the $38.8\mathrm{s}$ it takes to solve the CC-MPC-CBF optimization problem in one go. 
\section{Conclusion}\label{Conclusions}
In this paper, the MOCA problem in a stochastic scenario with unbounded uncertainties is addressed through the proposed CC-MPC-CBF approach, which combines MPC with chance-constrained CBFs to handle stochastic uncertainties and provides probabilistic guarantees on safety. A sequential implementation approach is also developed to improve the feasibility of the optimization of CC-MPC-CBF, which includes two sub-optimization problems, i.e., a standard MPC and a predictive safety filter. The effectiveness of the developed algorithms is demonstrated through numerous simulation results in a real-life MOCA example, which highlight their advantageous properties, such as robustness to large sensing uncertainties, high success rate for MOCA, fast computation speed, and feasibility for real-world applications. Overall, the proposed approach provides a promising solution to address the MOCA problem in stochastic systems with unbounded uncertainties.

\begin{table*}[tp]
\footnotesize
\centering
\caption{Feasibility of CC-MPC-CBF versus hyperparameter $\gamma$.}
\begin{tabular}{ccccccccccc}
  \hline
  Hyperparameter $\gamma$ & $\gamma=1$ & $\gamma=0.9$ & $\gamma=0.8$ & $\gamma=0.7$ & $\gamma=0.6$ & $\gamma=0.5$ \\
  \hline
  CC-MPC-CBF (\%) & 100 & 100 & 100 & 100 & 100 & 93 \\
  Infeasible at time $k$ & \textbackslash & \textbackslash & \textbackslash & \textbackslash & \textbackslash & $k=191$ \\
  \hline
\end{tabular}
\label{gamma_1}
\end{table*}
\begin{table*}[tp]
\footnotesize
\centering
\caption{(Continued) Feasibility of CC-MPC-CBF versus hyperparameter $\gamma$.}
\begin{tabular}{cccccccccc}
  \hline
  Hyperparameter $\gamma$ & $\gamma=0.4$ & $\gamma=0.3$ & $\gamma=0.2$ & $\gamma=0.1$ & $\gamma=0$ \\
  \hline
  CC-MPC-CBF (\%) & 87 & 61 & 42 & 25 & 13 \\
  Infeasible at time $k$ & $k=168$ & $k=133$ & $k=123$ & $k=77$ & $k=48$ \\
  \hline
\end{tabular}
\label{gamma_2}
\end{table*}
\begin{table*}[tp]
\footnotesize
\centering
\caption{Feasibility of CC-MPC-CBF with different prediction horizons $N$.}
\begin{tabular}{cccccccc}
  \hline
 Prediction horizons $N$ & $N=5$ & $N=10$ &$N=15$ & $N=20$ & $N=30$& $N=50$&$N=60$\\
CC-MPC-CBF (\%)&$100$&$100$&$93$&$88$&$60$&$0$&$0$\\
Infeasible at time $k$& \textbackslash & \textbackslash & $k=198$& $k=191$ & $k=133$ &$k=1$& $k=1$\\
  \hline
\end{tabular}
\label{Prediction_Hor}
\end{table*}

\bibliography{main}


\begin{thebibliography}{39}
\ifx \bisbn   \undefined \def \bisbn  #1{ISBN #1}\fi
\ifx \binits  \undefined \def \binits#1{#1}\fi
\ifx \bauthor  \undefined \def \bauthor#1{#1}\fi
\ifx \batitle  \undefined \def \batitle#1{#1}\fi
\ifx \bjtitle  \undefined \def \bjtitle#1{#1}\fi
\ifx \bvolume  \undefined \def \bvolume#1{\textbf{#1}}\fi
\ifx \byear  \undefined \def \byear#1{#1}\fi
\ifx \bissue  \undefined \def \bissue#1{#1}\fi
\ifx \bfpage  \undefined \def \bfpage#1{#1}\fi
\ifx \blpage  \undefined \def \blpage #1{#1}\fi
\ifx \burl  \undefined \def \burl#1{\textsf{#1}}\fi
\ifx \doiurl  \undefined \def \doiurl#1{\url{https://doi.org/#1}}\fi
\ifx \betal  \undefined \def \betal{\textit{et al.}}\fi
\ifx \binstitute  \undefined \def \binstitute#1{#1}\fi
\ifx \binstitutionaled  \undefined \def \binstitutionaled#1{#1}\fi
\ifx \bctitle  \undefined \def \bctitle#1{#1}\fi
\ifx \beditor  \undefined \def \beditor#1{#1}\fi
\ifx \bpublisher  \undefined \def \bpublisher#1{#1}\fi
\ifx \bbtitle  \undefined \def \bbtitle#1{#1}\fi
\ifx \bedition  \undefined \def \bedition#1{#1}\fi
\ifx \bseriesno  \undefined \def \bseriesno#1{#1}\fi
\ifx \blocation  \undefined \def \blocation#1{#1}\fi
\ifx \bsertitle  \undefined \def \bsertitle#1{#1}\fi
\ifx \bsnm \undefined \def \bsnm#1{#1}\fi
\ifx \bsuffix \undefined \def \bsuffix#1{#1}\fi
\ifx \bparticle \undefined \def \bparticle#1{#1}\fi
\ifx \barticle \undefined \def \barticle#1{#1}\fi
\bibcommenthead
\ifx \bconfdate \undefined \def \bconfdate #1{#1}\fi
\ifx \botherref \undefined \def \botherref #1{#1}\fi
\ifx \url \undefined \def \url#1{\textsf{#1}}\fi
\ifx \bchapter \undefined \def \bchapter#1{#1}\fi
\ifx \bbook \undefined \def \bbook#1{#1}\fi
\ifx \bcomment \undefined \def \bcomment#1{#1}\fi
\ifx \oauthor \undefined \def \oauthor#1{#1}\fi
\ifx \citeauthoryear \undefined \def \citeauthoryear#1{#1}\fi
\ifx \endbibitem  \undefined \def \endbibitem {}\fi
\ifx \bconflocation  \undefined \def \bconflocation#1{#1}\fi
\ifx \arxivurl  \undefined \def \arxivurl#1{\textsf{#1}}\fi
\csname PreBibitemsHook\endcsname

\bibitem[\protect\citeauthoryear{Murray et~al.}{2017}]{Manipulator_Control}
\begin{bbook}
\bauthor{\bsnm{Murray}, \binits{R.M.}},
\bauthor{\bsnm{Li}, \binits{Z.}},
\bauthor{\bsnm{Sastry}, \binits{S.S.}}:
\bbtitle{A Mathematical Introduction to Robotic Manipulation}.
\bpublisher{CRC press},
\blocation{Boca Raton, FL, USA}
(\byear{2017})
\end{bbook}
\endbibitem

\bibitem[\protect\citeauthoryear{Xu et~al.}{2017}]{Autonomous_Vehicles}
\begin{barticle}
\bauthor{\bsnm{Xu}, \binits{X.}},
\bauthor{\bsnm{Grizzle}, \binits{J.W.}},
\bauthor{\bsnm{Tabuada}, \binits{P.}},
\bauthor{\bsnm{Ames}, \binits{A.D.}}:
\batitle{Correctness guarantees for the composition of lane keeping and adaptive cruise control}.
\bjtitle{IEEE Transactions on Automation Science and Engineering}
\bvolume{15}(\bissue{3}),
\bfpage{1216}--\blpage{1229}
(\byear{2017})
\end{barticle}
\endbibitem

\bibitem[\protect\citeauthoryear{Zhao et~al.}{2017}]{Formation_Control}
\begin{barticle}
\bauthor{\bsnm{Zhao}, \binits{S.}},
\bauthor{\bsnm{Dimarogonas}, \binits{D.V.}},
\bauthor{\bsnm{Sun}, \binits{Z.}},
\bauthor{\bsnm{Bauso}, \binits{D.}}:
\batitle{A general approach to coordination control of mobile agents with motion constraints}.
\bjtitle{IEEE Transactions on Automatic Control}
\bvolume{63}(\bissue{5}),
\bfpage{1509}--\blpage{1516}
(\byear{2017})
\end{barticle}
\endbibitem

\bibitem[\protect\citeauthoryear{Alonso-Mora et~al.}{2015}]{Obstacle_Velocity}
\begin{barticle}
\bauthor{\bsnm{Alonso-Mora}, \binits{J.}},
\bauthor{\bsnm{Naegeli}, \binits{T.}},
\bauthor{\bsnm{Siegwart}, \binits{R.}},
\bauthor{\bsnm{Beardsley}, \binits{P.}}:
\batitle{Collision avoidance for aerial vehicles in multi-agent scenarios}.
\bjtitle{Autonomous Robots}
\bvolume{39},
\bfpage{101}--\blpage{121}
(\byear{2015})
\end{barticle}
\endbibitem

\bibitem[\protect\citeauthoryear{Mansouri et~al.}{2019}]{Potential_fields}
\begin{bchapter}
\bauthor{\bsnm{Mansouri}, \binits{S.S.}},
\bauthor{\bsnm{Karvelis}, \binits{P.}},
\bauthor{\bsnm{Kanellakis}, \binits{C.}},
\bauthor{\bsnm{Kominiak}, \binits{D.}},
\bauthor{\bsnm{Nikolakopoulos}, \binits{G.}}:
\bctitle{Vision-based {MAV} navigation in underground mine using convolutional neural network}.
In: \bbtitle{IECON 2019-45th Annual Conference of the IEEE Industrial Electronics Society},
vol. \bseriesno{1},
pp. \bfpage{750}--\blpage{755}
(\byear{2019}).
\bcomment{IEEE}
\end{bchapter}
\endbibitem

\bibitem[\protect\citeauthoryear{Zhu and Alonso-Mora}{2019}]{MPC_1}
\begin{barticle}
\bauthor{\bsnm{Zhu}, \binits{H.}},
\bauthor{\bsnm{Alonso-Mora}, \binits{J.}}:
\batitle{Chance-constrained collision avoidance for {MAV}s in dynamic environments}.
\bjtitle{IEEE Robotics and Automation Letters}
\bvolume{4}(\bissue{2}),
\bfpage{776}--\blpage{783}
(\byear{2019})
\end{barticle}
\endbibitem

\bibitem[\protect\citeauthoryear{Lindqvist et~al.}{2020}]{MPC_2}
\begin{barticle}
\bauthor{\bsnm{Lindqvist}, \binits{B.}},
\bauthor{\bsnm{Mansouri}, \binits{S.S.}},
\bauthor{\bsnm{Agha-mohammadi}, \binits{A.-a.}},
\bauthor{\bsnm{Nikolakopoulos}, \binits{G.}}:
\batitle{Nonlinear {MPC} for collision avoidance and control of {UAV}s with dynamic obstacles}.
\bjtitle{IEEE Robotics and Automation Letters}
\bvolume{5}(\bissue{4}),
\bfpage{6001}--\blpage{6008}
(\byear{2020})
\end{barticle}
\endbibitem

\bibitem[\protect\citeauthoryear{Zeng et~al.}{2021}]{MPC_CBF}
\begin{bchapter}
\bauthor{\bsnm{Zeng}, \binits{J.}},
\bauthor{\bsnm{Zhang}, \binits{B.}},
\bauthor{\bsnm{Sreenath}, \binits{K.}}:
\bctitle{Safety-critical model predictive control with discrete-time control barrier function}.
In: \bbtitle{2021 American Control Conference (ACC)},
pp. \bfpage{3882}--\blpage{3889}
(\byear{2021}).
\bcomment{IEEE}
\end{bchapter}
\endbibitem

\bibitem[\protect\citeauthoryear{Tulbure and Khatib}{2020}]{WeightedDist_IROS19}
\begin{bchapter}
\bauthor{\bsnm{Tulbure}, \binits{A.}},
\bauthor{\bsnm{Khatib}, \binits{O.}}:
\bctitle{Closing the loop: Real-time perception and control for robust collision avoidance with occluded obstacles}.
In: \bbtitle{2020 IEEE/RSJ International Conference on Intelligent Robots and Systems (IROS)},
pp. \bfpage{5700}--\blpage{5707}
(\byear{2020}).
\bcomment{IEEE}
\end{bchapter}
\endbibitem

\bibitem[\protect\citeauthoryear{Thirugnanam et~al.}{2022}]{MahalanobisCBF2022}
\begin{bchapter}
\bauthor{\bsnm{Thirugnanam}, \binits{A.}},
\bauthor{\bsnm{Zeng}, \binits{J.}},
\bauthor{\bsnm{Sreenath}, \binits{K.}}:
\bctitle{Duality-based convex optimization for real-time obstacle avoidance between polytopes with control barrier functions}.
In: \bbtitle{American Control Conference (ACC)},
pp. \bfpage{2239}--\blpage{2246}
(\byear{2022}).
\bcomment{Employs Mahalanobis-distance CBF formulation}.
\burl{https://arxiv.org/abs/2107.08360}
\end{bchapter}
\endbibitem

\bibitem[\protect\citeauthoryear{Chen et~al.}{2022}]{AnisotropicCBF_TRO22}
\begin{barticle}
\bauthor{\bsnm{Chen}, \binits{T.}},
\bauthor{\bsnm{Swann}, \binits{A.}},
\bauthor{\bsnm{Yu}, \binits{J.}},
\bauthor{\bsnm{Shorinwa}, \binits{O.}},
\bauthor{\bsnm{III}, \binits{M.K.}},
\bauthor{\bsnm{Schwager}, \binits{M.}}:
\batitle{Safer-splat: A control barrier function for safe navigation with online {G}aussian splatting maps}.
\bjtitle{IEEE Robotics and Automation Letters}
\bvolume{7}(\bissue{4}),
\bfpage{10659}--\blpage{10666}
(\byear{2022}).
\bcomment{Introduces anisotropic (ellipsoidal) CBFs}
\end{barticle}
\endbibitem

\bibitem[\protect\citeauthoryear{Xu et~al.}{2020}]{SpeedAwareSMC18}
\begin{barticle}
\bauthor{\bsnm{Xu}, \binits{T.}},
\bauthor{\bsnm{Zhang}, \binits{S.}},
\bauthor{\bsnm{Jiang}, \binits{Z.}},
\bauthor{\bsnm{Liu}, \binits{Z.}},
\bauthor{\bsnm{Cheng}, \binits{H.}}:
\batitle{Collision avoidance of high-speed obstacles for mobile robots via the maximum-speed-aware velocity obstacle method}.
\bjtitle{IEEE Access}
\bvolume{8},
\bfpage{138493}--\blpage{138507}
(\byear{2020}).
\bcomment{Original speed-aware formulation first appeared at IEEE SMC ’18}
\end{barticle}
\endbibitem

\bibitem[\protect\citeauthoryear{Shi et~al.}{2018}]{Faessler2018SpeedAware}
\begin{barticle}
\bauthor{\bsnm{Shi}, \binits{D.}},
\bauthor{\bsnm{Dassau}, \binits{E.}},
\bauthor{\bsnm{Doyle}, \binits{F.J.}}:
\batitle{Adaptive zone model predictive control of artificial pancreas based on glucose-and velocity-dependent control penalties}.
\bjtitle{IEEE Transactions on Biomedical Engineering}
\bvolume{66}(\bissue{4}),
\bfpage{1045}--\blpage{1054}
(\byear{2018})
\end{barticle}
\endbibitem

\bibitem[\protect\citeauthoryear{Emam et~al.}{2019}]{Robust_CBF1}
\begin{bchapter}
\bauthor{\bsnm{Emam}, \binits{Y.}},
\bauthor{\bsnm{Glotfelter}, \binits{P.}},
\bauthor{\bsnm{Egerstedt}, \binits{M.}}:
\bctitle{Robust barrier functions for a fully autonomous, remotely accessible swarm-robotics testbed}.
In: \bbtitle{2019 IEEE 58th Conference on Decision and Control (CDC)},
pp. \bfpage{3984}--\blpage{3990}
(\byear{2019}).
\bcomment{IEEE}
\end{bchapter}
\endbibitem

\bibitem[\protect\citeauthoryear{Jankovic}{2018}]{Robust_CBF2}
\begin{barticle}
\bauthor{\bsnm{Jankovic}, \binits{M.}}:
\batitle{Robust control barrier functions for constrained stabilization of nonlinear systems}.
\bjtitle{Automatica}
\bvolume{96},
\bfpage{359}--\blpage{367}
(\byear{2018})
\end{barticle}
\endbibitem

\bibitem[\protect\citeauthoryear{Kolathaya and Ames}{2018}]{ISSf}
\begin{barticle}
\bauthor{\bsnm{Kolathaya}, \binits{S.}},
\bauthor{\bsnm{Ames}, \binits{A.D.}}:
\batitle{Input-to-state safety with control barrier functions}.
\bjtitle{IEEE Control Systems Letters}
\bvolume{3}(\bissue{1}),
\bfpage{108}--\blpage{113}
(\byear{2018})
\end{barticle}
\endbibitem

\bibitem[\protect\citeauthoryear{Sakhdari et~al.}{2017}]{Tube_Robust_MPC_CBF}
\begin{bchapter}
\bauthor{\bsnm{Sakhdari}, \binits{B.}},
\bauthor{\bsnm{Shahrivar}, \binits{E.M.}},
\bauthor{\bsnm{Azad}, \binits{N.L.}}:
\bctitle{Robust tube-based {MPC} for automotive adaptive cruise control design}.
In: \bbtitle{2017 IEEE 20th International Conference on Intelligent Transportation Systems (ITSC)},
pp. \bfpage{1}--\blpage{6}
(\byear{2017}).
\bcomment{IEEE}
\end{bchapter}
\endbibitem

\bibitem[\protect\citeauthoryear{Taylor and Ames}{2020}]{adactive_cbf}
\begin{bchapter}
\bauthor{\bsnm{Taylor}, \binits{A.J.}},
\bauthor{\bsnm{Ames}, \binits{A.D.}}:
\bctitle{Adaptive safety with control barrier functions}.
In: \bbtitle{2020 American Control Conference (ACC)},
pp. \bfpage{1399}--\blpage{1405}
(\byear{2020}).
\bcomment{IEEE}
\end{bchapter}
\endbibitem

\bibitem[\protect\citeauthoryear{Xiao et~al.}{2021}]{xiao2021adaptive}
\begin{barticle}
\bauthor{\bsnm{Xiao}, \binits{W.}},
\bauthor{\bsnm{Belta}, \binits{C.}},
\bauthor{\bsnm{Cassandras}, \binits{C.G.}}:
\batitle{Adaptive control barrier functions}.
\bjtitle{IEEE Transactions on Automatic Control}
\bvolume{67}(\bissue{5}),
\bfpage{2267}--\blpage{2281}
(\byear{2021})
\end{barticle}
\endbibitem

\bibitem[\protect\citeauthoryear{Ohnishi et~al.}{2019}]{Data_driven_CBF}
\begin{barticle}
\bauthor{\bsnm{Ohnishi}, \binits{M.}},
\bauthor{\bsnm{Wang}, \binits{L.}},
\bauthor{\bsnm{Notomista}, \binits{G.}},
\bauthor{\bsnm{Egerstedt}, \binits{M.}}:
\batitle{Barrier-certified adaptive reinforcement learning with applications to brushbot navigation}.
\bjtitle{IEEE Transactions on Robotics}
\bvolume{35}(\bissue{5}),
\bfpage{1186}--\blpage{1205}
(\byear{2019})
\end{barticle}
\endbibitem

\bibitem[\protect\citeauthoryear{Aali and Liu}{2024}]{GP_CBF_Bayesian}
\begin{bchapter}
\bauthor{\bsnm{Aali}, \binits{M.}},
\bauthor{\bsnm{Liu}, \binits{J.}}:
\bctitle{Learning high-order control barrier functions for safety-critical control with {G}aussian processes}.
In: \bbtitle{2024 American Control Conference (ACC)},
pp. \bfpage{1}--\blpage{6}
(\byear{2024}).
\bcomment{IEEE}
\end{bchapter}
\endbibitem

\bibitem[\protect\citeauthoryear{Clark}{2021}]{Stochastic_CBF}
\begin{barticle}
\bauthor{\bsnm{Clark}, \binits{A.}}:
\batitle{Control barrier functions for stochastic systems}.
\bjtitle{Automatica}
\bvolume{130},
\bfpage{109688}
(\byear{2021})
\end{barticle}
\endbibitem

\bibitem[\protect\citeauthoryear{Emam et~al.}{2022}]{RL_CBF_SRL}
\begin{botherref}
\oauthor{\bsnm{Emam}, \binits{Y.}},
\oauthor{\bsnm{Notomista}, \binits{G.}},
\oauthor{\bsnm{Glotfelter}, \binits{P.}},
\oauthor{\bsnm{Kira}, \binits{Z.}},
\oauthor{\bsnm{Egerstedt}, \binits{M.}}:
Safe reinforcement learning using robust control barrier functions.
IEEE Robotics and Automation Letters
(2022)
\end{botherref}
\endbibitem

\bibitem[\protect\citeauthoryear{Yin et~al.}{2023}]{MPSC_CBF}
\begin{barticle}
\bauthor{\bsnm{Yin}, \binits{J.}},
\bauthor{\bsnm{Dawson}, \binits{C.}},
\bauthor{\bsnm{Fan}, \binits{C.}},
\bauthor{\bsnm{Tsiotras}, \binits{P.}}:
\batitle{Shield model predictive path integral: A computationally efficient robust mpc method using control barrier functions}.
\bjtitle{IEEE Robotics and Automation Letters}
\bvolume{8}(\bissue{11}),
\bfpage{7106}--\blpage{7113}
(\byear{2023})
\end{barticle}
\endbibitem

\bibitem[\protect\citeauthoryear{Schwarm and Nikolaou}{1999}]{schwarm1999chance}
\begin{barticle}
\bauthor{\bsnm{Schwarm}, \binits{A.T.}},
\bauthor{\bsnm{Nikolaou}, \binits{M.}}:
\batitle{Chance-constrained model predictive control}.
\bjtitle{AIChE Journal}
\bvolume{45}(\bissue{8}),
\bfpage{1743}--\blpage{1752}
(\byear{1999})
\end{barticle}
\endbibitem

\bibitem[\protect\citeauthoryear{Wang et~al.}{2024}]{wang2024stochastic}
\begin{bchapter}
\bauthor{\bsnm{Wang}, \binits{Y.}},
\bauthor{\bsnm{Shen}, \binits{X.}},
\bauthor{\bsnm{Qian}, \binits{H.}}:
\bctitle{Stochastic model predictive control with probabilistic control barrier functions and smooth sample-based approximation}.
In: \bbtitle{2024 IEEE 63rd Conference on Decision and Control (CDC)},
pp. \bfpage{4798}--\blpage{4803}
(\byear{2024}).
\bcomment{IEEE}
\end{bchapter}
\endbibitem

\bibitem[\protect\citeauthoryear{Liu et~al.}{2025}]{liu2025flexible}
\begin{botherref}
\oauthor{\bsnm{Liu}, \binits{J.}},
\oauthor{\bsnm{Yang}, \binits{J.}},
\oauthor{\bsnm{Mao}, \binits{J.}},
\oauthor{\bsnm{Zhu}, \binits{T.}},
\oauthor{\bsnm{Xie}, \binits{Q.}},
\oauthor{\bsnm{Li}, \binits{Y.}},
\oauthor{\bsnm{Wang}, \binits{X.}},
\oauthor{\bsnm{Li}, \binits{S.}}:
Flexible active safety motion control for robotic obstacle avoidance: A cbf-guided mpc approach.
IEEE Robotics and Automation Letters
(2025)
\end{botherref}
\endbibitem

\bibitem[\protect\citeauthoryear{Agrawal and Sreenath}{2017}]{Discrete_CBF}
\begin{bchapter}
\bauthor{\bsnm{Agrawal}, \binits{A.}},
\bauthor{\bsnm{Sreenath}, \binits{K.}}:
\bctitle{Discrete control barrier functions for safety-critical control of discrete systems with application to bipedal robot navigation.}
In: \bbtitle{Robotics: Science and Systems},
vol. \bseriesno{13}
(\byear{2017}).
\bcomment{Cambridge, MA, USA}
\end{bchapter}
\endbibitem

\bibitem[\protect\citeauthoryear{Grimmett and Stirzaker}{2020}]{Approximate_Gaussian}
\begin{bbook}
\bauthor{\bsnm{Grimmett}, \binits{G.}},
\bauthor{\bsnm{Stirzaker}, \binits{D.}}:
\bbtitle{Probability and Random Processes},
\bedition{4th} edn.
\bpublisher{Oxford University Press},
\blocation{Oxford, UK}
(\byear{2020})
\end{bbook}
\endbibitem

\bibitem[\protect\citeauthoryear{Rencher and Schaalje}{2008}]{Mean_Covariance_Results}
\begin{bbook}
\bauthor{\bsnm{Rencher}, \binits{A.C.}},
\bauthor{\bsnm{Schaalje}, \binits{G.B.}}:
\bbtitle{Linear Models in Statistics},
\bedition{2nd} edn.
\bpublisher{John Wiley \& Sons},
\blocation{Hoboken, NJ}
(\byear{2008})
\end{bbook}
\endbibitem

\bibitem[\protect\citeauthoryear{Blackmore et~al.}{2011}]{Chance_constrained_lemma}
\begin{barticle}
\bauthor{\bsnm{Blackmore}, \binits{L.}},
\bauthor{\bsnm{Ono}, \binits{M.}},
\bauthor{\bsnm{Williams}, \binits{B.C.}}:
\batitle{Chance-constrained optimal path planning with obstacles}.
\bjtitle{IEEE Transactions on Robotics}
\bvolume{27}(\bissue{6}),
\bfpage{1080}--\blpage{1094}
(\byear{2011})
\end{barticle}
\endbibitem

\bibitem[\protect\citeauthoryear{Zeng et~al.}{2021}]{CC-MPC-CLF-CBF}
\begin{bchapter}
\bauthor{\bsnm{Zeng}, \binits{J.}},
\bauthor{\bsnm{Li}, \binits{Z.}},
\bauthor{\bsnm{Sreenath}, \binits{K.}}:
\bctitle{Enhancing feasibility and safety of nonlinear model predictive control with discrete-time control barrier functions}.
In: \bbtitle{2021 60th IEEE Conference on Decision and Control (CDC)},
pp. \bfpage{6137}--\blpage{6144}
(\byear{2021}).
\bcomment{IEEE}
\end{bchapter}
\endbibitem

\bibitem[\protect\citeauthoryear{Zhu and Alonso-Mora}{2019}]{Chance_Constraied}
\begin{barticle}
\bauthor{\bsnm{Zhu}, \binits{H.}},
\bauthor{\bsnm{Alonso-Mora}, \binits{J.}}:
\batitle{Chance-constrained collision avoidance for {MAVs} in dynamic environments}.
\bjtitle{IEEE Robotics and Automation Letters}
\bvolume{4}(\bissue{2}),
\bfpage{776}--\blpage{783}
(\byear{2019})
\end{barticle}
\endbibitem

\bibitem[\protect\citeauthoryear{Magni et~al.}{1997}]{MPC_Relaxation}
\begin{bbook}
\bauthor{\bsnm{Magni}, \binits{J.-F.}},
\bauthor{\bsnm{Bennani}, \binits{S.}},
\bauthor{\bsnm{Terlouw}, \binits{J.}}:
\bbtitle{Robust Flight Control: a Design Challenge}.
\bsertitle{Lecture Notes in Control and Information Sciences},
vol. \bseriesno{110}.
\bpublisher{Springer},
\blocation{London, UK}
(\byear{1997})
\end{bbook}
\endbibitem

\bibitem[\protect\citeauthoryear{Morari and Lee}{1999}]{MPC_Literature1}
\begin{barticle}
\bauthor{\bsnm{Morari}, \binits{M.}},
\bauthor{\bsnm{Lee}, \binits{J.H.}}:
\batitle{Model predictive control: past, present and future}.
\bjtitle{Computers \& Chemical Engineering}
\bvolume{23}(\bissue{4-5}),
\bfpage{667}--\blpage{682}
(\byear{1999})
\end{barticle}
\endbibitem

\bibitem[\protect\citeauthoryear{Darby and Nikolaou}{2012}]{MPC_Literature2}
\begin{barticle}
\bauthor{\bsnm{Darby}, \binits{M.L.}},
\bauthor{\bsnm{Nikolaou}, \binits{M.}}:
\batitle{{MPC}: Current practice and challenges}.
\bjtitle{Control Engineering Practice}
\bvolume{20}(\bissue{4}),
\bfpage{328}--\blpage{342}
(\byear{2012})
\end{barticle}
\endbibitem

\bibitem[\protect\citeauthoryear{Ames et~al.}{2019}]{Safety_filter}
\begin{bchapter}
\bauthor{\bsnm{Ames}, \binits{A.D.}},
\bauthor{\bsnm{Coogan}, \binits{S.}},
\bauthor{\bsnm{Egerstedt}, \binits{M.}},
\bauthor{\bsnm{Notomista}, \binits{G.}},
\bauthor{\bsnm{Sreenath}, \binits{K.}},
\bauthor{\bsnm{Tabuada}, \binits{P.}}:
\bctitle{Control barrier functions: Theory and applications}.
In: \bbtitle{2019 18th European Control Conference (ECC)},
pp. \bfpage{3420}--\blpage{3431}
(\byear{2019}).
\bcomment{IEEE}
\end{bchapter}
\endbibitem

\bibitem[\protect\citeauthoryear{Drgo{\v{n}}a et~al.}{2020}]{MPC_Feasibility}
\begin{barticle}
\bauthor{\bsnm{Drgo{\v{n}}a}, \binits{J.}},
\bauthor{\bsnm{Arroyo}, \binits{J.}},
\bauthor{\bsnm{Figueroa}, \binits{I.C.}},
\bauthor{\bsnm{Blum}, \binits{D.}},
\bauthor{\bsnm{Arendt}, \binits{K.}},
\bauthor{\bsnm{Kim}, \binits{D.}},
\bauthor{\bsnm{Oll{\'e}}, \binits{E.P.}},
\bauthor{\bsnm{Oravec}, \binits{J.}},
\bauthor{\bsnm{Wetter}, \binits{M.}},
\bauthor{\bsnm{Vrabie}, \binits{D.L.}}, \betal:
\batitle{All you need to know about model predictive control for buildings}.
\bjtitle{Annual Reviews in Control}
\bvolume{50},
\bfpage{190}--\blpage{232}
(\byear{2020})
\end{barticle}
\endbibitem

\bibitem[\protect\citeauthoryear{Wabersich and Zeilinger}{2022}]{predictive_safe_filter}
\begin{botherref}
\oauthor{\bsnm{Wabersich}, \binits{K.P.}},
\oauthor{\bsnm{Zeilinger}, \binits{M.N.}}:
Predictive control barrier functions: Enhanced safety mechanisms for learning-based control.
IEEE Transactions on Automatic Control
(2022)
\end{botherref}
\endbibitem

\end{thebibliography}
\end{document}